\documentclass[11pt]{article}

\usepackage[margin=1in,letterpaper]{geometry}
\usepackage[T1]{fontenc}
\usepackage[sc]{mathpazo}
\usepackage{array}
\usepackage{microtype}
\usepackage{amsmath}
\usepackage[dvipsnames]{xcolor}
\usepackage{amssymb}
\usepackage{amsthm}
\usepackage{bm,float}
\usepackage{dsfont}
\usepackage{authblk}
\usepackage{fullpage,caption,wrapfig}
\usepackage{comment}
\usepackage{mathtools}
\usepackage[shortlabels,inline]{enumitem}
\usepackage{bbm}
\usepackage{booktabs}
\usepackage{makecell}
\usepackage{diagbox}
\usepackage{hyperref}
\definecolor{pku-red}{RGB}{139,0,18}
\usepackage{nag,tikz}
\usetikzlibrary{calc}
\usepackage{tablefootnote}
\usetikzlibrary{decorations.pathreplacing}
\usetikzlibrary{graphs, positioning, quotes, shapes.geometric}
\usetikzlibrary{arrows.meta}%

\usepackage{multicol}
\usepackage[scaled]{helvet} 
\usepackage{thmtools} 

\usepackage{float}

\usepackage{comment}
\usepackage{multirow}

\usepackage[normalem]{ulem}

\usepackage{crossreftools}
\usepackage{etoolbox}
\makeatletter
\patchcmd{\@bibitem}{\ignorespaces}{\label{bib-#1}\ignorespaces}{}{}
\makeatother

\usepackage{algorithm} 
\usepackage{algorithmic} 

\hypersetup{
    colorlinks=true,
    linkcolor=pku-red,     
    urlcolor=violet,
    citecolor=BlueViolet,
    pdffitwindow=true,
}\usepackage[capitalize, nameinlink]{cleveref}

\newenvironment{interposition}[1][name]
{\medskip\noindent\ignorespaces\textbf{#1.}}{\medskip\noindent\ignorespacesafterend}

\allowdisplaybreaks[3]
\newcommand{\R}{\mathbb{R}}
\newcommand{\T}[1]{{#1}^T}

\theoremstyle{plain}
\newtheorem{theorem}{Theorem}
\newtheorem{lemma}{Lemma}
\newtheorem{proposition}{Proposition}
\newtheorem{corollary}{Corollary}

\theoremstyle{definition}
\newtheorem{definition}{Definition}
\newtheorem{example}{Example}

\theoremstyle{remark}
\newtheorem{remark}{Remark}

\DeclareMathOperator*{\suppmax}{suppmax}
\DeclareMathOperator*{\suppmin}{suppmin}
\DeclareMathOperator*{\argmin}{argmin}
\DeclareMathOperator*{\supp}{supp}
\DeclareMathOperator*{\poly}{poly}

\newcommand\NEn{Nash equilibrium}
\newcommand\NE{Nash equilibrium }
\newcommand\SPn{stationary point}
\newcommand\SP{stationary point }
\newcommand{\FP}{fictitious play }

\newcommand\lhysays[1]{\textcolor{blue}{#1}}

\newcommand\AdjFi{(x_{\mathrm{TS}}, y_{\mathrm{TS}})}
\newcommand\AdjSe{(x_{\mathrm{MB}}, y_{\mathrm{MB}})}
\newcommand\AdjTh{(x_{\mathrm{IL}}, y_{\mathrm{IL}})}

\makeatletter
\newenvironment{breakablealgorithm}
  {
   \begin{center}
     \refstepcounter{algorithm}
     \hrule height.8pt depth0pt \kern2pt
     \renewcommand{\caption}[2][\relax]{
       {\raggedright\textbf{\ALG@name~\thealgorithm} ##2\par}
       \ifx\relax##1\relax 
         \addcontentsline{loa}{algorithm}{\protect\numberline{\thealgorithm}##2}
       \else 
         \addcontentsline{loa}{algorithm}{\protect\numberline{\thealgorithm}##1}
       \fi
       \kern2pt\hrule\kern2pt
     }
  }{
     \kern2pt\hrule\relax
   \end{center}
  }
\makeatother

\title{On Tightness of Tsaknakis-Spirakis Descent Methods for Approximate Nash Equilibria}

\author[1]{Zhaohua Chen\thanks{\href{mailto:chenzhaohua@pku.edu.cn}{chenzhaohua@pku.edu.cn}.}}

\author[1]{Xiaotie Deng\thanks{\href{mailto:xiaotie@pku.edu.cn}{xiaotie@pku.edu.cn}.}}

\author[2]{Wenhan Huang\thanks{\href{mailto:rowdark@sjtu.edu.cn}{rowdark@sjtu.edu.cn}.}}

\author[1]{Hanyu Li\thanks{\href{mailto:lhydave@pku.edu.cn}{lhydave@pku.edu.cn}, the main technical contributor.}}

\author[3]{Yuhao Li\thanks{\href{mailto:yuhaoli@cs.columbia.edu}{yuhaoli@cs.columbia.edu}.}}

\affil[1]{Peking University}
\affil[2]{Shanghai Jiao Tong University}
\affil[3]{Columbia University}

\begin{document}

\maketitle
\begin{abstract}
This article explores the minimum approximation ratio for Nash equilibrium in bi-matrix games, focusing on the Tsaknakis and Spirakis (TS) methods. The previous SOTA, TS algorithm, achieved an approximation ratio of 0.3393, but efforts to improve the analysis of the TS algorithm have been unsuccessful. This work demonstrates that the bound of 0.3393 is tight for the TS algorithm and presents a theoretical worst-case analysis. A condition for identifying tight instances is provided, along with a generator. While most generated instances are unstable, indicating potential improvements, stable instances exist where perturbations cannot enhance the 0.3393 bound. Other approximate algorithms, such as regret-matching and fictitious play, achieve better ratios on these instances. The generated instances can serve as benchmarks for approximate Nash equilibrium algorithms. The article also mentions progress in the TS algorithm, achieving an approximation ratio of 1/3, which can be further studied using the presented techniques.
\end{abstract}

\section{Introduction}\label{sec:intro}
Nash equilibria characterize a static status of players in which each player cannot gain more utilities by changing her current strategy. 
Ever since Nash proved the existence of such an equilibrium \cite{nash1951non}, it has become the fundamental concept in non-cooperative game theory and economics.
In view of computer science, computing Nash equilibria rises in great importance from computational complexity theory \cite{chen2009settling, DBLP:journals/siamcomp/DaskalakisGP09, rubinstein2016settling}, algorithmic game theory \cite{DBLP:journals/csr/KoutsoupiasP09, DBLP:journals/geb/NisanR01, DBLP:journals/jacm/RoughgardenT02}, and learning theory \cite{DBLP:journals/jmlr/BlumM07, DBLP:books/daglib/0016248, DBLP:conf/icml/HuW98}.
It has been shown that \NE computing lies in the complexity class PPAD introduced by Papadimitriou \cite{DBLP:journals/jcss/Papadimitriou94}. 
The PPAD-completeness results have been established for 3NASH (finding an approximate solution of Nash equilibria for at-least-three-player games) by Daskalakis, Goldberg, and Papadimitriou \cite{DBLP:journals/siamcomp/DaskalakisGP09}, and for 2NASH (finding an exact/approximate solution of Nash equilibria for two-player games) by Chen, Deng, and Teng \cite{chen2009settling}. 
It is well believed that PPAD-complete problems can hardly have a polynomial-time algorithm.
These completeness results thus lead to a great many efforts to find an $\epsilon$-approximate \NE (see the formal definition in Section~\ref{sec:prelim}) in polynomial time for some small constant $\epsilon>0$. 

Early works by Kontogiannis et al. \cite{kontogiannis2009polynomial} and Daskalakis et al. \cite{daskalakis2009note} introduced polynomial-time algorithms to reach an approximation ratio of $\epsilon = 3/4$ and $\epsilon = 1/2$, respectively. 
Their algorithms are based on searching strategies with small supports. Conitzer \cite{conitzer2009approximation} also showed that the well-known fictitious play algorithm \cite{brown1951iterative} gives a $1/2$-approximate \NE within constant rounds, matching Feder et al.'s lower-bound result \cite{feder2007approximating}. 
Subsequently, Daskalakis et al. \cite{daskalakis2007progress} gave an algorithm with an approximation ratio of $0.38$ by enumerating arbitrarily large supports. 
The same result was achieved by Czumaj et al. \cite{czumaj2019distributed} with a totally different approach by finding the Nash equilibria of two zero-sum games and further making a convex combination between the solution and the corresponding best responses. 
Bosse et al. \cite{bosse2010new} provided another algorithm based on the previous work of Kontogiannis and Spirakis \cite{kontogiannis2007efficient}, and their algorithm reaches a $0.36$-approximation ratio. 
Concurrent with these works, Tsaknakis and Spirakis \cite{tsaknakis2008optimization} established the best previously known approximation ratio of $0.3393$. 
The basic idea of the Tsaknakis and Spirakis (TS) algorithm is to directly optimize the approximation ratio itself by a descent procedure, after which a further adjustment is conducted by making a convex combination between the solution produced by the optimization and the corresponding best responses.

The original paper of Tsaknakis and Spirakis \cite{tsaknakis2008optimization} proved that the approximation ratio of the algorithm is at most $0.3393$. 
However, it leaves open whether $0.3393$ is tight for the algorithm. 
In the literature, the experimental performance of the algorithm is far better than $0.3393$ \cite{tsaknakis2008performance}. 
The worst approximation ratio in experiments reported ahead of our paper is provided in Fearnley et al. \cite{fearnley2015empirical}, where the TS algorithm on a game finds a $0.3385$-approximate \NEn.

In this work, exploring the descent procedure in the TS algorithm, we present a delicate analysis of the lower bound of the TS algorithm, which is illustrated with several images of the worst cases. 
It provides a full understanding of the worst cases of the TS algorithm. 
The analysis allows us to prove that $0.3393$ is indeed the tight bound for the TS algorithm by providing a bimatrix game instance. 
Subsequently, we solve the open problem regarding the well-followed TS algorithm. 
Furthermore, we characterize all game instances that are able to attain the tight bound.  
Based on this characterization, we identify a region of payoff matrices that the game instances generated are precisely tight instances. 
This identification allows us to propose a generator of tight instances and conduct various experiments on the generated instances.

Despite the tightness of $0.3393$ for the TS algorithm, our extensive experiments show that it is rather difficult to reach a $0.3393$ bound on generated instances in practice by brute-force enumerations. 
Such results imply that the experimental bound is \emph{usually} inconsistent with the theoretical bound. 
To figure out the reasons for such a gap, we further explore the stability of generated instances. 
Our empirical studies show that most generated instances of large sizes are \emph{unstable}. 
The word ``unstable'' means that a small perturbation near a $0.3393$ solution would make the TS algorithm find another faraway solution with a much better approximation ratio. 
With the game size growing large, the probability of finding a stable tight instance plunges and even vanishes. 
These results help to understand the gap: Even if a tight instance is met, the TS algorithm usually escapes the $0.3393$ solution and reaches a far better approximation ratio. 
Based on these results, we give a time-saving and effective suggestion on the practical usage of the TS algorithm in \Cref{sec:experi}.

We also use the generated game instances to measure the performances of the Czumaj et al.'s algorithm~\cite{czumaj2019distributed}, the regret-matching algorithm in online learning~\cite{greenwald2006bounds}, and the \FP algorithm~\cite{brown1951iterative}. 
The regret-matching algorithm and the \FP algorithm perform well on these instances. Interestingly, the algorithm of Czumaj et al. always reaches an empirical approximation ratio of $0.3393$ on generated game instances, implying that the tight instance generator for the TS algorithm also makes a totally different algorithm perform poorly. 
Such results show that our instances generated against the TS algorithm serve as a necessary benchmark in the design and analysis of approximate \NE algorithms.

\textbf{Subsequent Work.} After the conference version \cite{DBLP:conf/sagt/ChenDHLL21} of the presented paper, very recently, the work by Deligkas, Fasoulakis, and Markakis \cite{DBLP:journals/corr/abs-2204-11525} provides a polynomial-time algorithm\footnote{We call it DFM algorithm for short.} computing a $1/3$-approximate Nash equilibrium. 
Their algorithm is also based on the same descent procedure but makes a more delicate adjustment.
We show the generality of our study to provide a proof using similar techniques that the approximation ratio $1/3$ of their algorithm is tight.

This paper is organized as follows. 
In \Cref{sec:prelim}, we introduce the basic definitions and notations that we use throughout the paper. 
In \Cref{sec:algo}, we restate the TS algorithm~\cite{tsaknakis2008optimization} and propose two other auxiliary methods which help to analyze the original algorithm. 
In \Cref{sec:TI}, we dive into the worst-case analysis, show what a tight instance looks like, and then prove by a game instance that $0.3393$ is indeed the tight bound for the TS algorithm.
Further, we characterize all tight instances and present a generator that outputs tight game instances in \Cref{sec:gen-TI}. 
With similar techniques, we show that the upper bound $1/3$ of the DFM algorithm is indeed tight by presenting matching-bound instances in \Cref{sec:1/3-tight}. 
We conduct extensive experiments to reveal the properties of the stationary points and compare the descent methods with other approximate \NE algorithms in \Cref{sec:experi}. 
At last, we present several open problems raised from our study in \Cref{sec:disc}.

\section{Definitions and Notations}\label{sec:prelim}

We focus on finding an approximate \NE in normal-form games with two players. 
We use $R_{m\times n}$ and $C_{m\times n}$ to denote the payoff matrices of the row player and column player, where the row player and the column player have $m$ and $n$ strategies, respectively. 
Furthermore, we suppose that both $R$ and $C$ are normalized so that all their entries belong to $[0, 1]$. 
In fact, concerning Nash equilibria, any game is equivalent to a normalized game with appropriate shifting and scaling of both payoff matrices.

For two vectors $u$ and $v$ with the same length, we say $u \geq v$ if each entry of $u$ is greater than or equal to the corresponding entry of $v$. 
Meanwhile, let us denote by $e_k$ a $k$-dimension vector with all entries equal to $1$. 
We use a probability vector to define either player's behavior, which describes the probability that a player chooses any pure strategy to play. 
Specifically, the row player's strategy and the column player's strategy lie in $\Delta_m$ and $\Delta_n$, respectively, where
\begin{align*}
    \Delta_m &= \{x\in\R^m:x\geq 0, \T{x}e_m = 1\}, \\
    \Delta_n &= \{y\in\R^n:y\geq 0, \T{y}e_n = 1\}.
\end{align*}

For a strategy pair $(x, y) \in \Delta_m \times \Delta^n$, we call it an \emph{$\epsilon$-approximate \NEn}, if for any $x'\in\Delta_m$, $y'\in\Delta_n$, the following inequalities hold:
\begin{align*}
    \T{x'}Ry\leq \T{x}Ry+\epsilon,\\
    \T{x}Cy'\leq \T{x}Cy+\epsilon.
\end{align*}

Therefore, a \NE is an $\epsilon$-approximate \NE with $\epsilon = 0$.

To simplify our further discussion, for any probability vector $u$, we use
\[\supp(u)=\{i:u_i>0\},\]
to denote the support of $u$, and
\begin{align*}
    \suppmax(u) &= \{i:\forall j,\ u_i\geq u_j\},\\
    \suppmin(u) &= \{i:\forall j,\ u_i\leq u_j\},
\end{align*}
to denote the index set of all entries equal to the maximum/minimum entry of vector $u$.

At last, we use $\max(u)$ to denote the value of the maximal entry of vector $u$, and $\max\limits_S(u)$ to denote the value of the maximal entry of vector $u$ confined in the index set $S$.

\section{Algorithms}\label{sec:algo}

In this section, we first restate the TS algorithm~\cite{tsaknakis2008optimization}, and then propose two auxiliary adjusting methods, which help to analyze the bound of the TS algorithm. 

The TS algorithm formulates the approximate \NE problem into an optimization problem. Specifically, we define the following functions:
\begin{align*}
    f_R(x,y) &:= \max(Ry)-\T{x}Ry, \\
    f_C(x,y) &:= \max(\T{C}x)-\T{x}Cy, \\
    f(x, y) &:= \max\left\{f_R(x, y), f_C(x, y)\right\}.
\end{align*}
The goal is to minimize $f(x,y)$ over $\Delta_m\times\Delta_n$. 

The relationship between the above function $f$ and approximate \NE is as follows. Given strategy pair $(x,y) \in \Delta_m\times\Delta_n$, $f_R(x,y)$ and $f_C(x,y)$ are the respective maximum deviations of row player and column player. By definition, $(x,y)$ is an $\epsilon$-approximate \NE if and only if $f(x,y)\leq\epsilon$. In other words, as long as we obtain a point with $f$ value no greater than $\epsilon$, an $\epsilon$-approximate \NE is reached.

The idea of the TS algorithm is to find a \SP (to be defined in \Cref{def:sp}) of the objective function $f$ by a descent procedure and make a further adjustment from the \SPn  \footnote{We will see in \Cref{remark:1/2-sp} that finding a \SP is not enough to reach a good approximation ratio; therefore the adjustment step is necessary.}. 
To give the formal definition of \SPn s, we need to define the Dini directional derivative of $f$ as follows:

\begin{definition}\label{def:dir-deri}
Given $(x, y), (x',y')\in\Delta_m\times\Delta_n$, the \emph{Dini directional derivative} \cite{Demyanov2010} of $(x, y)$ in direction $(x'-x,y'-y)$ is
\[Df(x,y,x',y'):=\lim_{\theta\to 0+}\frac{1}{\theta}\left(f(x+\theta(x'-x),y+\theta(y'-y))-f(x,y)\right).\]  
$Df_R(x,y,x',y')$ and $Df_C(x,y,x',y')$ are defined similarly with respect to $f_R$ and $f_C$.
\end{definition}

\begin{remark}
Note that the notion of $Df(x,y,x',y')$ in \Cref{def:dir-deri}  is not the directional derivative we usually consider. 
The latter should be defined as 
\[\lim_{\theta\to 0+}\frac{1}{\theta}\left(f(x+\theta\frac{x'-x}{\Vert x'-x\Vert},y+\theta\frac{y'-y}{\Vert y-y'\Vert})-f(x,y)\right).\]

We give an example to show the difference. Let $(x'',y'')=(x,y)+1/2(x'-x,y'-y)$. It is clear that $(x''-x,y''-y)$ represents the same direction as $(x'-x,y'-y)$. However, 
\begin{align*}
Df(x,y,x'',y'')&=\lim_{\theta\to 0+}\frac{1}{\theta}\left(f(x+\theta(x''-x),y+\theta(y''-y))-f(x,y)\right)\\
&=\frac{1}{2}\lim_{\theta\to 0+}\frac{1}{\theta/2}\left(f(x+\theta/2(x'-x),y+\theta/2(y'-y))-f(x,y)\right)\\
&=\frac{1}{2}Df(x,y,x',y').
\end{align*}

So even when we consider points in the same direction, their $Df$ values will be different by a multiple. Why do we not normalize the direction vectors but keep using Dini derivatives? We will see soon that this definition shares good properties: It is easier to analyze and compute the steepest direction.
\end{remark}

Now we present the definition of \SPn s. 
\begin{definition}\label{def:sp}
$(x,y)\in\Delta_m\times\Delta_n$ is a \emph{\SPn} if and only if for any $(x',y')\in\Delta_m\times\Delta_n$,
\[Df(x,y,x',y')\geq 0.\]
\end{definition}

We use a descent procedure to find a \SPn. The descent procedure is presented in \Cref{app:DP}. Due to time and precision limit, we cannot expect the descent procedure always finds an exact stationary point. Instead, we seek a $\delta$-stationary point as in the following definition.
\begin{definition}
Given $\delta\geq 0$, $(x,y)\in\Delta_m\times\Delta_n$ is a $\delta$-\emph{\SPn} if and only if 
\[f_R(x,y)=f_C(x,y)\]
and for any $(x',y')\in\Delta_m\times\Delta_n$,
\[Df(x,y,x',y')\geq -\delta.\]
\end{definition}

It has already been proved that the procedure runs in the polynomial-time of precision $\delta$ to find a $\delta$-\SPn~\cite{tsaknakis2008performance}.

To better deal with $Df(x, y, x', y')$, we give an explicit form of $Df$. Detailed calculations are presented in \Cref{app:MissingCalculations}. In \cite{tsaknakis2008optimization}, they have provided alternative characterizations of $Df$. Below we provide a similar development, but emphasize its min-max structure, which is utilized in a series of proofs later.  For now we only care about cases when $f_R(x,y)=f_C(x,y)$ (which is a necessary condition for a stationary point to be proved in \Cref{prop:construct-sp}). Let $S_C(x):=\suppmax(\T{C}x)$, $S_R(y):=\suppmax(Ry)$. Then the formula is 
\begin{align*}
    Df(x,y,x',y')&=\max\{Df_R(x,y,x',y'),Df_C(x,y,x',y')\}\\
    &=\max\{T_1(x,y,x',y'),T_2(x,y,x',y')\}-f(x,y),
\end{align*}
where
\begin{align*}
T_1(x,y,x',y')&=\max_{S_R(y)}(Ry')-\T{x'}Ry-\T{x}Ry'+\T{x}Ry,\\ T_2(x,y,x',y')&=\max_{S_C(x)}(\T{C}x')-\T{x'}Cy-\T{x}Cy'+\T{x}Cy.
\end{align*}

A key component of $Df$ is $\max\{T_1,T_2\}$, for which several maximum operators are applied. To smoothen these maximum operations, we introduce linear convex combinations via $\rho, w$ and $z$:
\begin{align*}
T(x,y,x',y',\rho,w,z):=&\rho(\T{w}Ry'-\T{x}Ry'-\T{x'}Ry+\T{x}Ry)\\
         &\quad+(1-\rho)(\T{x'}Cz-\T{x}Cy'-\T{x'}Cy+\T{x}Cy),
\end{align*}
where $\rho\in[0, 1]$, $w\in\Delta_m$, $\supp(w)\subseteq S_R(y)$, $z\in\Delta_n$, $\supp(z)\subseteq S_C(x)$.\footnote{Throughout the paper,  we require that $(x, y), (x',y')\in\Delta_m\times\Delta_n$, and $\rho\in[0, 1]$, $w\in\Delta_m$, $\supp(w)\subseteq S_R(y)$, $z\in\Delta_n$, $\supp(z)\subseteq S_C(x)$. These restrictions are omitted afterward for fluency of presentation.} When $f_R(x,y)=f_C(x,y)$, we have the following identities
\begin{align*}
    &\max_{\rho, w, z}T(x,y,x',y',\rho,w,z)\\
    =&\max_\rho\max_{w,z} T(x,y,x',y',\rho,w,z)\\
    =&\max_\rho (\rho\max_w( \T{w}Ry'-\T{x}Ry'-\T{x'}Ry+\T{x}Ry)\\
         \phantom{=}&\quad+(1-\rho)\max_z(\T{x'}Cz-\T{x}Cy'-\T{x'}Cy+\T{x}Cy))\\
    =&\max_\rho\left(\rho T_1(x,y,x',y')+(1-\rho)T_2(x,y,x',y')\right)\\
    =&\max\left\{T_1(x,y,x',y'), T_2(x,y,x',y')\right\}.
\end{align*}
Thus
\[Df(x, y, x', y') = \max_{\rho, w, z}T(x,y,x',y',\rho,w,z) - f(x,y).\]

Now define
\[V(x, y) := \min_{x',y'}\max_{\rho,w,z} T(x,y,x',y',\rho,w,z).\]
By \Cref{def:sp}, $(x, y)$ is a \SP if and only if $V(x, y) \geq f(x, y)$. Further, by substituting $x', y'$ with $x,y$, we have $V(x, y) \leq \max_{\rho,w,z} T(x, y, x, y, \rho, w, z)$, which is the same as
$f(x, y)$ because $T_1(x,y,x,y)=f_R(x,y)$ and $T_2(x,y,x,y)=f_C(x,y)$ by their definitions.

Therefore, we have the following proposition.

\begin{proposition}\label{prop:sp-equi}
$(x, y)$ is a stationary point if and only if 
\[V(x, y) = f_R(x, y) = f_C(x, y).\]
\end{proposition}

By making some variations on $T$ (see \Cref{app:MissingCalculations} for detail calculations), we can show that $T$ is a bilinear form in $(\rho w,(1-\rho)z)$ and $(x',y')$, i.e., $T$ is equal to
\[(\rho \T{w}, (1-\rho)\T{z})\,G(x,y)\begin{pmatrix}y'\\x'\end{pmatrix}\]
for some $(m+n)\times(m+n)$ matrix $G(x,y)$. Thus by applying von Neumann's minimax theorem~\cite{neumann1928theorie}, we have
\begin{proposition}\label{prop:dual}
\[V(x,y) = \max_{\rho,w,z}\min_{x',y'} T(x,y,x',y',\rho,w,z), \]
and there exist $\rho_0,w_0,z_0$ such that
\[V(x,y) = \min_{x',y'} T(x,y,x',y',\rho_0,w_0,z_0).\]
\end{proposition}

We call the tuple $(\rho_0,w_0,z_0)$ a \emph{dual solution} as it can be calculated by dual linear programming. See \Cref{app:MissingCalculations} for the calculations in detail. 

In the following context, we fix $(x^*, y^*)$ to denote a \SP and use $(\rho^*,w^*,z^*)$ to denote the corresponding dual solution about $G(x^*,y^*)$. 

As we will see in \Cref{remark:1/2-sp}, a \SP may only achieve an approximation ratio of $1/2$ in the worst case. To find a better solution, we adjust the \SP to another point lying in the following square:
\[\Lambda := \{(\alpha w^* + (1 - \alpha)x^*, \beta z^* + (1 - \beta)y^*): \alpha, \beta \in [0, 1]\}.\]

Different adjustments on $\Lambda$ derive different algorithms to find an approximate \NEn. 
We present three of these methods below, of which the first one is the solution by the TS algorithm, and the other two are for the sake of analysis in \Cref{sec:TI}. For simplicity of the presentation, we define the following two subsets of the boundary of $\Lambda$:
\begin{align*}
    \Gamma_1 &:= \{(\alpha x^*+(1-\alpha)w^*,y^*):\alpha\in[0,1]\}\cup\{(x^*,\beta y^*+(1-\beta)z^*:\beta\in[0,1]\}, \\
    \Gamma_2 &:= \{(\alpha x^*+(1-\alpha)w^*,z^*):\alpha\in[0,1]\}\cup\{(w^*,\beta y^*+(1-\beta)z^*:\beta\in[0,1]\}.
\end{align*}

\begin{enumerate}[fullwidth, label=\textbf{Method \arabic*}. ]
    \item \begin{interposition}[Method in the original TS algorithm~\cite{tsaknakis2008optimization}]
The first method is the original adjustment given by~\cite{tsaknakis2008optimization} (known as \emph{the TS algorithm} in literature). Define the quantities
\begin{align*}
    \lambda&:=\min_{y':\supp(y')\subseteq S_C(x^*)}\{\T{(w^*-x^*)}Ry'\},\\
    \mu&:=\min_{x':\supp(x')\subseteq S_R(y^*)}\{\T{x'}C(z^*-y^*)\}.
\end{align*}
The adjusted strategy pair is
\begin{align*}
\AdjFi:=\begin{cases}\left(\frac{1}{1+\lambda-\mu}w^*+\frac{\lambda-\mu}{1+\lambda-\mu}x^*,z^*\right),&\lambda\geq\mu,\\
\left(w^*,\frac{1}{1+\mu-\lambda}z^*+\frac{\mu-\lambda}{1+\mu-\lambda}y^*\right),&\lambda<\mu.\end{cases}
\end{align*}
\end{interposition}
\item \begin{interposition}[Minimum point on $\Gamma_2$]
For the second method, define
\begin{align*}
    \alpha^*&:=\argmin_{\alpha\in[0,1]} f(\alpha w^*+(1-\alpha)x^*,z^*),\\
    \beta^*&:=\argmin_{\beta\in[0,1]} f(w^*,\beta z^*+(1-\beta)y^*).
\end{align*}
From a geometric view, our goal is to find the minimum point of $f$ on $\Gamma_2$. The strategy pair given by the second method is 
\begin{align*}
  \AdjSe:=\begin{cases}\left(\alpha^*w^*+(1-\alpha^*)x^*,z^*\right),& f_C(w^*,z^*)\geq f_R(w^*,z^*),\\
\left(w^*,\beta^*z^*+(1-\beta^*)y^*\right),& f_C(w^*,z^*)<f_R(w^*,z^*).
\end{cases}  
\end{align*}
When we write the strategy into such two different cases, it is not trivial to see $\AdjSe$ is indeed the minimum point of $f$ on $\Gamma_2$. We will prove this fact in \Cref{lemma:best-on-boundary}.
\end{interposition}
\item \begin{interposition}[Intersection point of linear bound of $f_R$ and $f_C$ on $\Gamma_2$]
As we will see later, $\AdjSe$ always behaves no worse than $\AdjFi$ theoretically. However, it is rather hard to quantitatively analyze the exact approximation ratio of $\AdjSe$ given in the second method. Therefore, we propose a third adjustment method.  It is not hard to see directly from definitions that $f_R(x,y)$, $f_C(x,y)$ and $f(x,y)$ are all convex and linear-piecewise functions with either $x$ or $y$ fixed. Therefore, on the boundary of $\Lambda$, they can be bounded by linear functions. Formally, for $0 \leq p, q \leq 1$, we have
\begin{align}
f_R(p w^*+(1-p)x^*,z^*)&= (f_R(w^*,z^*)-f_R(x^*,z^*))p+f_R(x^*,z^*),\label{derive1}\\
f_C(p w^*+(1-p)x^*,z^*)&\leq f_C(w^*,z^*)p;\label{derive2}\\
f_C(w^*,q z^*+(1-q) y^*)&= (f_C(w^*,x^*)-f_C(w^*,y^*))q+f_C(w^*,y^*),\label{derive3}\\
f_R(w^*,q z^*+(1-q) y^*)&\leq f_R(w^*,z^*)q.\label{derive4}
\end{align}
Taking the minimum of terms on the right hand sides of \cref{derive1} and \cref{derive2}, \cref{derive3} and \cref{derive4} respectively, i.e.,
\[p^*\in\argmin_{p\in\{0,1\}}\min\left\{(f_R(w^*,z^*)-f_R(x^*,z^*))p+f_R(x^*,z^*), f_C(w^*,z^*)p\right\},\]
\[q^*\in\argmin_{q\in\{0,1\}}\min\left\{(f_C(w^*,x^*)-f_C(w^*,y^*))q+f_C(w^*,y^*), f_R(w^*,z^*)q\right\},\]
we derive the following quantities\footnote{The denominator of $p^*$ or $q^*$ may be zero. In this case, we simply define $p^*$ or $q^*$ to be $0$.}
\begin{align*}
p^*&:=\frac{f_R(x^*,z^*)}{f_R(x^*,z^*)+f_C(w^*,z^*)-f_R(w^*,z^*)},\\
q^*&:=\frac{f_C(w^*,y^*)}{f_C(w^*,y^*)+f_R(w^*,z^*)-f_C(w^*,z^*)}.
\end{align*}
The adjusted strategy pair is now defined as
\begin{align*}
  \AdjTh:=\begin{cases}\left(p^*w^*+(1-p^*)x^*,z^*\right),& f_C(w^*,z^*)\geq f_R(w^*,z^*),\\
\left(w^*,q^*z^*+(1-q^*)y^*\right),& f_C(w^*,z^*)<f_R(w^*,z^*).
\end{cases}  
\end{align*}

In \Cref{sec:TI}, we will see that
it is easy to analyze this strategy pair quantitatively, and it is a key auxiliary structure that brings about a thorough worst-case analysis of the TS algorithm.
\end{interposition}
\end{enumerate}

We remark that the outcome of all these three methods can be calculated in polynomial-time of $m$ and $n$.

\section{A Tight Instance for All Three Methods}\label{sec:TI}

We now show the tight bound of the TS algorithm that we presented in the previous section, with the help of two auxiliary adjustment methods proposed in \Cref{sec:algo}. \cite{tsaknakis2008optimization} has shown that the TS algorithm gives an approximation ratio of no greater than $b\approx 0.3393$. 
In this section, we construct a game on which the TS algorithm attains the tight bound $b \approx0.3393$. In detail, the payoff matrices of the game are presented in \cref{ex:tight-bound}, where $b\approx 0.3393$ is the tight bound, $\lambda_0 \approx 0.582523$ and $\mu_0 \approx 0.812815$ are the real numbers to be derived in \Cref{lemma:ineq-for-est}. The game attains the tight bound $b\approx 0.3393$ at the \SP $x^*=y^*=\T{(1,0,0)}$ with dual solution $\rho^*=\mu_0/(\lambda_0+\mu_0)$, $w^*=z^*=\T{(0,0,1)}$. Additionally, the bound stays $b\approx 0.3393$ for this game even when we try to find the minimum point of $f$ on the entire space of square $\Lambda$.

\begin{align}\label{ex:tight-bound}
    R=\begin{pmatrix}0.1&0&0\\0.1+b&1&1\\0.1+b&\lambda_0&\lambda_0\end{pmatrix},\qquad
C=\begin{pmatrix}0.1&0.1+b&0.1+b\\0&1&\mu_0\\0&1&\mu_0\end{pmatrix}.
\end{align}
The formal statement of this result is presented in the following \Cref{thm:worst-case-exist}.
\begin{theorem}[Tightness of the generalized TS algorithm]\label{thm:worst-case-exist}
There exists a game such that for some stationary point $(x^*,y^*)$ with dual solution $(\rho^*,w^*,z^*)$, 
\[b=f(x^*,y^*)=f\AdjTh = f\AdjSe\leq f(\alpha w^*+(1-\alpha)x^*,\beta z^*+(1-\beta)y^*)\]
holds for any $\alpha,\beta\in[0,1]$.
\end{theorem}

The proof of \Cref{thm:worst-case-exist} is done by verifying the tight instance \cref{ex:tight-bound} above. However, it is not direct to verify this fact. More importantly, it is far from triviality how this tight instance comes about. Below, we present the thread of our idea by a series of lemmas and propositions that help us find the tight instance \cref{ex:tight-bound}.


We sketch our preparation into three steps. First, we give an equivalent condition of the \SP in \Cref{prop:construct-sp}, which makes it easier to construct payoff matrices with a given \SP and its corresponding dual solution. Second, we draw figures of functions $f_R$ and $f_C$ on $\Lambda$ and subsequently reveal the relationship among the three adjusting strategy pairs 
presented in \Cref{sec:algo}. Finally, we present some estimations over $f$ and show when these estimations are exactly tight. 

During the preparations (or more precisely, attempts), we have found more accurate constraints for tight instances. A parameterization method arises naturally, and thus the tight instances were found by trying very few cases.\footnote{Such a procedure can be written into an algorithm, as we will show in \Cref{sec:gen-TI}.}

We have seen that stationary points are closely  linked to von Neumann minimax theorem. We here utilize it again but in a more delicate way. Specifically, the following proposition shows how to construct payoff matrices with a given stationary point $(x^*,y^*)$ and its dual solution $(\rho^*, w^*, z^*)$.
\begin{proposition}\label{prop:construct-sp}
Let
\begin{align*}
    A(\rho,y,z)&:=-\rho Ry+(1-\rho)C(z - y),\\
    B(\rho,x,w)&:=\rho \T{R}(w - x)-(1-\rho)\T{C}x.
\end{align*}
Then $(x^*,y^*)$ is a stationary point if and only if $f_R(x^*,y^*)=f_C(x^*,y^*)$ and there exist $\rho^*, w^*, z^*$ such that
\begin{align}
\supp(x^*)&\subset\suppmin(A(\rho^*,y^*,z^*)),\label{cond:sp-1}\\
\supp(y^*)&\subset\suppmin(B(\rho^*,x^*,w^*)).\label{cond:sp-2}
\end{align}
\end{proposition}

\begin{proof}
First, we show that $f_R(x^*,y^*)=f_C(x^*,y^*)$ is the necessary condition for $(x^*,y^*)$ to be a stationary point. We prove the contraposition. Suppose that $f_R(x^*,y^*)>f_C(x^*,y^*)$, then we have $f_R(x^*,y^*)>0$, which implies that $\max(Ry^*)>\T{x^*}Ry^*$. Therefore $\supp(x^*) \nsubseteq \suppmax(Ry^*)$.

Suppose without loss of generality that $1\in \suppmax(Ry^*)$, $2\notin\suppmax(Ry^*)$ and $2\in\supp(x^*)$. Let $E := \T{(1,-1,0,\cdots,0)}\in\R^m$. For sufficiently small $\theta_0>0$, we have $(x^*+\theta_0 E,y^*)\in \Delta_m\times\Delta_n$ and $f_R(x^*+\theta_0 E,y^*)>f_C(x^*+\theta_0 E,y^*)$. One can verify that
\begin{align*}
    Df(x^*,y^*,x^*+\theta_0 E,y^*)&=Df_R(x^*,y^*,x^*+\theta_0 E,y^*)\\
    &=-\theta_0 \T{E}Ry^*<0.
\end{align*}
Therefore $(x^*, y^*)$ is not a stationary point. The case of $f_C(x^*,y^*)>f_R(x^*,y^*)$ is similar.

Next, we prove that under the condition that $f_R(x^*,y^*)=f_C(x^*,y^*)$, $V(x^*,y^*)=f(x^*,y^*)$ if and only if \cref{cond:sp-1} and \cref{cond:sp-2} hold for some $\rho^*,w^*,z^*$.

Suppose $f(x^*,y^*)=V(x^*,y^*)$, by \Cref{prop:dual}, there exist $\rho^*,w^*,z^*$ such that
\[f(x^*,y^*) = \min_{x',y'} T(x^*,y^*,x',y',\rho^*,w^*,z^*).\]
Rewrite $T$ as
\begin{align*}
    T(x^*,y^*,x',y',\rho^*,w^*,z^*)&=\T{x'}A(\rho^*,y^*,z^*)+\T{B(\rho^*,x^*,w^*)}y' \\ 
    &+ \rho^* \T{x^*}Ry^* + (1 - \rho^*)\T{x^*}Cy^*.
\end{align*}
Notice that 
\begin{align*}
T(x^*,y^*,x^*,y^*,\rho^*,w^*,z^*) &=f(x^*,y^*)\\
&=\min_{x',y'} T(x^*,y^*,x',y',\rho^*,w^*,z^*).
\end{align*}
Therefore, 
\begin{align*}
    \supp(x^*) \subseteq \suppmin A(\rho^*,y^*,z^*), \\
    \supp(y^*) \subseteq \suppmin B(\rho^*,x^*,w^*)
\end{align*}
must hold.

Now suppose \cref{cond:sp-1} and \cref{cond:sp-2} hold. Similarly, we have
\begin{align*}
    V(x^*, y^*) &= \min_{x',y'} T(x^*,y^*,x',y',\rho^*,w^*,z^*) \\
    &= T(x^*,y^*,x^*,y^*,\rho^*,w^*,z^*) = f(x^*,y^*).
\end{align*}
\end{proof}

Note that given stationary point $(x^*,y^*)$ and dual solution $(\rho^*,w^*,z^*)$, we can restrict $R$ and $C$ by simple linear constraints: Finding $(R,C)$ becomes a problem of solving linear equations and linear inequalities.


Now we turn to the second step, i.e., plotting the figure of $f_R$ and $f_C$ on the rectangle $\Lambda$ in general cases.
To avoid burden notations, we define
\[F_I(\alpha,\beta):=f_I(\alpha w^*+(1-\alpha)x^*,\beta z^*+(1-\beta)y^*), I\in\{R, C\}, \alpha, \beta\in[0, 1].\]
Alternatively, we will show the figures of $F_R(\alpha,\beta)$ and $F_C(\alpha,\beta)$.
An instance is presented in \Cref{fig:origin} and \Cref{fig:full_f}.

\begin{figure}[ht]
    \centering
    \includegraphics[scale=0.65]{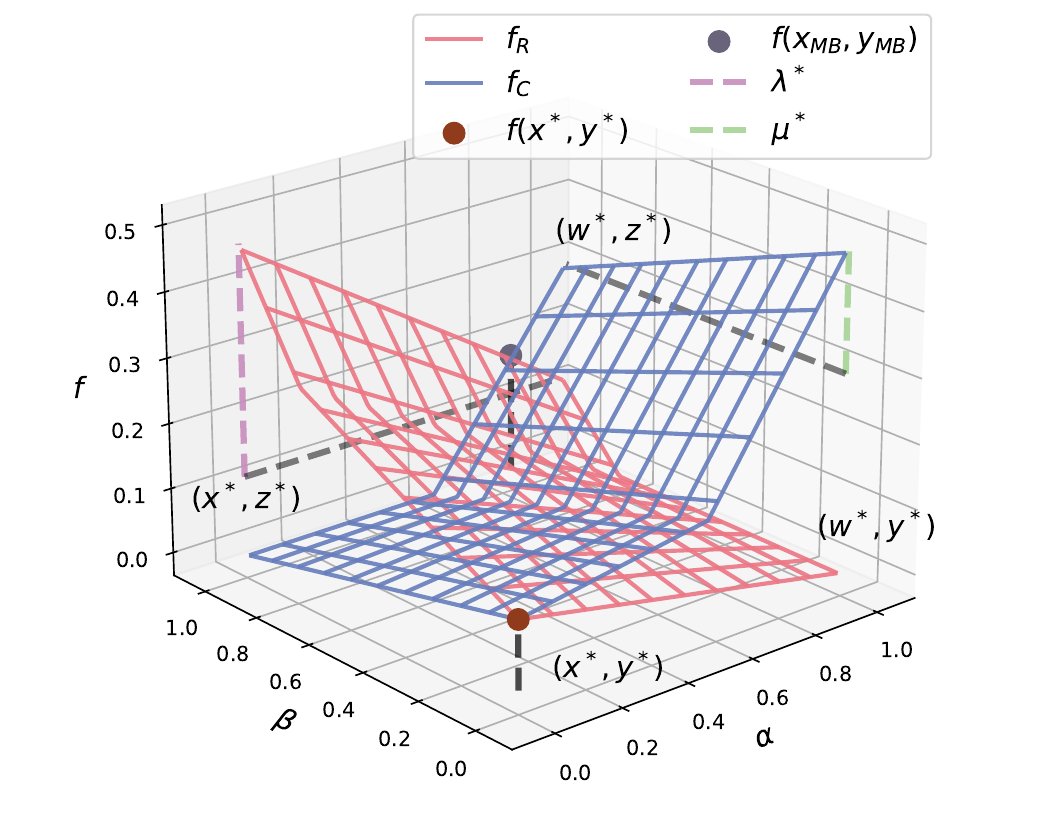}
    \caption{$f_R$ and $f_C$ on rectangle $\Lambda$. We also tag the critical points $(x^*,y^*)$ and $\AdjSe$ and the height differences $\lambda^*$ and $\mu^*$. Text $(x^*,y^*)$ in the figure means $\alpha=0, \beta=0$. Similarly, $(x^*,z^*)$ means $\alpha=0, \beta=1$; $(w^*,y^*)$ means $\alpha=1, \beta=0$; $(w^*,z^*)$ means $\alpha=1, \beta=1$.}
    \label{fig:origin}
\end{figure}

\begin{figure}[ht]
    \centering
    \includegraphics[scale=0.65]{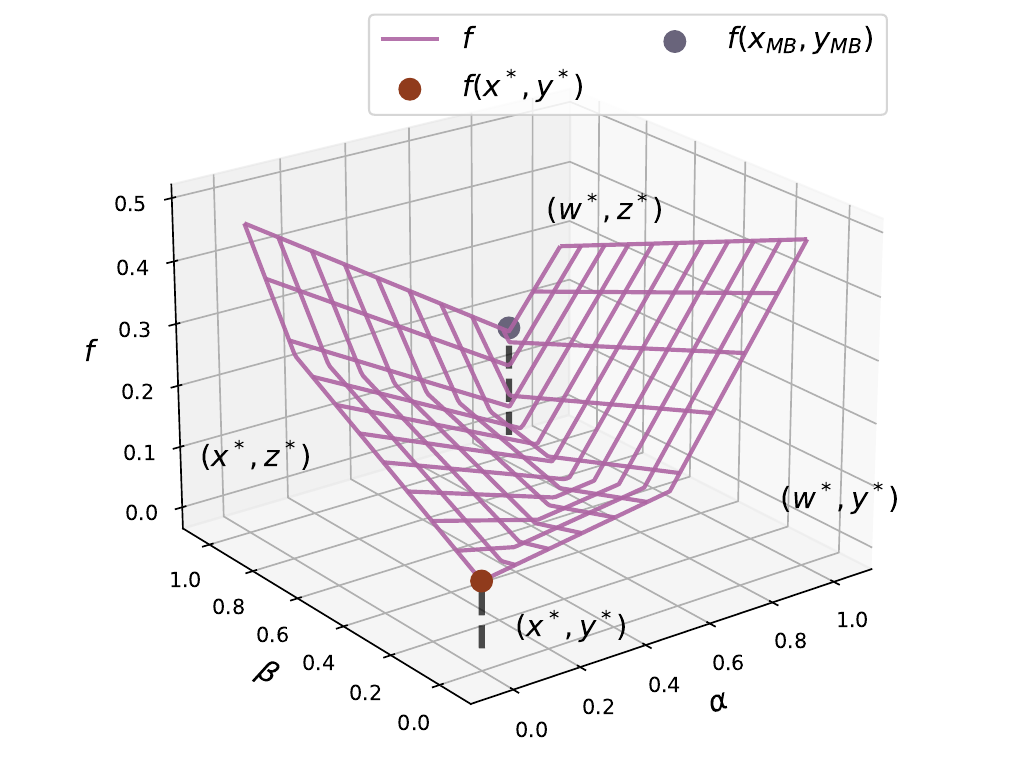}
    \caption{$f$ on rectangle $\Lambda$. We also tag the critical points $(x^*,y^*)$ and $\AdjSe$.}
    \label{fig:full_f}
\end{figure}

To understand why $f_R$ and $f_C$ have such a shape, we first define quantities $\lambda^*$ and $\mu^*$ as follows, which have both geometric and algebraic meanings. Let
\begin{align*}
\lambda^* &:= \T{(w^*-x^*)}Rz^*=f_R(x^*,z^*)-f_R(w^*,z^*)=F_R(0,1)-F_R(1,1),\\
  \mu^* &:=\T{w^*}C(z^*-y^*)=f_C(w^*,y^*)-f_C(w^*,z^*)=F_C(1,0)-F_C(1,1).
\end{align*}

The vertical dashed colored lines in \Cref{fig:origin} show the geometric meaning of these quantities: they are height differences. The following lemma shows that $\lambda^*$ and $\mu^*$ are always nonnegative height differences shown in \Cref{fig:origin}.
\begin{lemma}\label{lemma:lambda*-mu*}
If $\rho^*\in(0,1)$, then $\lambda^*,\mu^*\in[0,1]$. And if $\rho^*\in\{0,1\}$, \SP $(x^*,y^*)$ is a Nash equilibrium. 
\end{lemma}

\begin{proof}
We have $\lambda^*,\mu^*\leq 1$ as all entries of $R$ and $C$ belong to $[0, 1]$. Suppose $\rho^*\in(0,1)$. By \Cref{prop:sp-equi}, $0\leq f(x^*,y^*)\leq T(x^*,y^*,x^*,z^*,\rho^*,w^*,z^*)=\rho^*\lambda^*$, therefore $\lambda^*\geq 0$. Similarly, $0\leq f(x^*,y^*)\leq T(x^*,y^*,w^*,y^*,\rho^*,w^*,z^*)=(1-\rho^*)\mu^*$, therefore $\mu^*\geq 0$.

Suppose $\rho^*\in\{0,1\}$. By the previous inequalities, $0\leq f(x^*,y^*)\leq\min\{\rho^*\lambda^*, (1-\rho^*)\mu^*\}=0$. Therefore $(x^*,y^*)$ is a Nash equilibrium.
\end{proof}

Below we always assume $\rho^*\in(0,1)$, since otherwise a \NE $(x^*,y^*)$ is found, which does not match our goal of finding a tight instance.

The following lemma shows how $F_R$ and $F_C$ look like in the section when either $\alpha$ or $\beta$ is fixed. The colored solid lines in \Cref{fig:origin} present the image of this lemma.
\begin{lemma}\label{lemma:monotone-convex}
The following two statements hold:
\begin{enumerate}
\item Given $\beta$, $F_C(\alpha,\beta)$ is an increasing, convex and piecewise-linear function of $\alpha$; $F_R(\alpha,\beta)$ is a decreasing and linear function of $\alpha$.
\item Given $\alpha$, $F_R(\alpha,\beta)$ is an increasing and  convex, piecewise-linear function of $\beta$; $F_C(\alpha,\beta)$ is a decreasing and linear function of $\beta$.
\end{enumerate}
\end{lemma}

\begin{proof}
We only prove the first statement here and the second one is symmetric. Let $x_\alpha := \alpha w^*+(1-\alpha)x^*$, $y_\beta := \beta z^*+(1-\beta)y^*$.

Notice that 
\[F_C(\alpha,\beta) = \max(\T{C}x_\alpha) - \T{x_\alpha}Cy_\beta, \]
therefore is 
convex and piecewise-linear in $\alpha$ with fixed $\beta$. A similar argument holds for $F_R(\alpha,\beta)$. We then show the increasing property for $F_R$. In fact, 
\[F_R(\alpha, \beta) = \max(Ry_\beta) - \T{x_\alpha}Ry_\beta, \]
therefore is linear in $\alpha$ with fixed $\beta$. Further by $\supp(w^*) \subseteq S_R(y^*)$ and \Cref{lemma:lambda*-mu*}, we have
\begin{align*}
    F_R(0, \beta) - F_R(1, \beta) &= (1 - \beta)\T{(w^* - x^*)}Ry^* + \beta\T{(w^* - x^*)}Rz^* \\
    &\geq \beta\lambda^* \geq 0,
\end{align*}
which shows that $F_R(\alpha, \beta)$ is decreasing with fixed $\beta$.

At last, to prove that $F_C(\alpha,\beta)$ is increasing in $\alpha$, by convexity, it suffices to show that $Df_C(x^*,y_\beta,w^*,y_\beta)\geq 0$. Note that $Df_R(x^*,y^*,w^*,y^*) \leq 0$. By the definition of \SPn, we must have
\begin{align}
    0&\leq Df(x^*,y^*,w^*,y^*)=Df_C(x^*,y^*,w^*,y^*)\notag\\
    &=\max_{S_C(x^*)}(\T{C}w^*)-\max(\T{C}x^*)+\T{(x^*-w^*)}Cy^*.\label{prove:Df1geq0}
\end{align}
Notice that $f_C(x^*,z^*)=0$ and $f_C(x,z^*)\geq 0$ for all valid $x$, so
\begin{align}
    0&\leq Df_C(x^*,z^*,w^*,z^*)\notag\\
    &=\max_{S_C(x^*)}(\T{C}w^*)-\max(\T{C}x^*)+\T{(x^*-w^*)}Cz^*.\label{prove:Df2geq0}
\end{align}
Combining \cref{prove:Df1geq0} and \cref{prove:Df2geq0}, we have
\begin{align*}
    Df_C(x^*,y_\beta,w^*,y_\beta)&=\max_{S_C(x^*)}(\T{C}w^*)-\max(\T{C}x^*)+\T{(x^*-w^*)}Cy_\beta\\
    &=\beta Df_C(x^*,z^*,w^*,z^*)+(1-\beta)Df_C(x^*,y^*,w^*,z^*)\geq 0.
\end{align*}
\end{proof}

Recall that the second adjustment method yields the strategy pair $\AdjSe$. We have the following lemma indicating that $(x^*,y^*)$ and $\AdjSe$ are the minimum points on the boundary of $\Lambda$. They are the colored dots in \Cref{fig:origin}.

\begin{lemma}\label{lemma:best-on-boundary}
The following two statements hold:
\begin{enumerate}
    \item $(x^*,y^*)$ is the minimum point of $f$ on $\Gamma_1=\{(\alpha x^*+(1-\alpha)w^*,y^*):\alpha\in[0,1]\}\cup\{(x^*,\beta y^*+(1-\beta)z^*:\beta\in[0,1]\}$.
    \item $\AdjSe$ is the minimum point of $f$ on $\Gamma_2=\{(\alpha x^*+(1-\alpha)w^*,z^*):\alpha\in[0,1]\}\cup\{(w^*,\beta y^*+(1-\beta)z^*:\beta\in[0,1]\}$.
\end{enumerate}
\end{lemma}

\begin{proof}
Let $x_\alpha:=\alpha w^*+(1-\alpha)x^*$, $y_\beta:=\beta z^*+(1-\beta)y^*$. 
For the first part, by \Cref{prop:sp-equi}, $f_R(x^*,y^*)=f_C(x^*,y^*)=f(x^*,y^*)$. Meanwhile, \Cref{lemma:monotone-convex} shows that $f_C(x_\alpha,y^*)$ is an increasing function of $\alpha$, and $f_R(x_\alpha,y^*)$ is an decreasing function of $\alpha$, therefore $f(x_\alpha,y^*)=f_C(x_\alpha,y^*)\geq f_C(x^*,y^*)=f(x^*,y^*)$. Similarly, $f(x^*,y_\beta)=f_R(x^*,y_\beta)\geq f(x^*,y^*)$. As a result, $(x^*,y^*)$ is the minimum point of $f$ on $\Gamma_1$.

For the second part, suppose $f_C(w^*,z^*)\geq f_R(w^*,z^*)$. Again, by \Cref{lemma:monotone-convex} and a similar argument, $f(w^*,y_{\beta})=f_C(w^*,y_{\beta})\geq f_C(w^*,z^*)=f(w^*,z^*)\geq f(x_{\alpha^*},z^*)$. Therefore $\AdjSe=(x_{\alpha^*},z^*)$ is the minimum point on $\Gamma_2$. A similar argument holds for the case $f_R(w^*,z^*) > f_C(w^*,z^*)$.
\end{proof}

From the above lemmas, it is clear how \Cref{fig:origin} comes about. Now we turn to find tight instances. To take a further step, it is not enough only to plot a sketch. We need a quantitative analysis, in other words, calculating the exact heights of $F_R$ and $F_C$ in \Cref{fig:origin}. We then try to make the lowest point of the figure as high as possible, which may lead to a tight instance. We first do such a process on the boundary of $\Lambda$, denoted by $\partial\Lambda$. Then we show that it naturally leads to the worst-case analysis on the whole square $\Lambda$ as well.

As \Cref{lemma:best-on-boundary} suggests, either $(x^*,y^*)$ or $\AdjSe$ is a minimum point on $\partial\Lambda$, depending on their $f$ values. $(x^*,y^*)$ is a \SPn; thus it owns many properties owing to its dual LP structure, which can be utilized to estimate $f(x^*,y^*)$. The main barrier, however, is to analyze $\AdjSe$, whose position in \Cref{fig:origin} seems random, let alone its $f$ value. 
We give an estimable upper bound of $f\AdjSe$ by developing a shifted point along the boundary of square $\Lambda$ as follows.

Recall that \Cref{lemma:monotone-convex} shows when $\alpha$ or $\beta$ is fixed, the figure of $F_C$ and $F_R$ is convex. For every section, the figure of $F_R$ and $F_C$ becomes a convex curve. Fixing the two endpoints of the curve, we stretch the convex curve into a line, which gives an upper bound of $F_R$ or $F_C$. After stretching every section, the figure is stretched to a smooth surface with every section linear. Since we only care about $\AdjSe$, by the definition of $\AdjSe$, we only need to stretch $F_R$ or $F_C$ so that $\AdjSe$ is lifted. Such a procedure is shown in \Cref{fig:tostretch} and the result is shown in \Cref{fig:stretched}.

\begin{figure}[ht]
    \centering
    \includegraphics[scale=0.55]{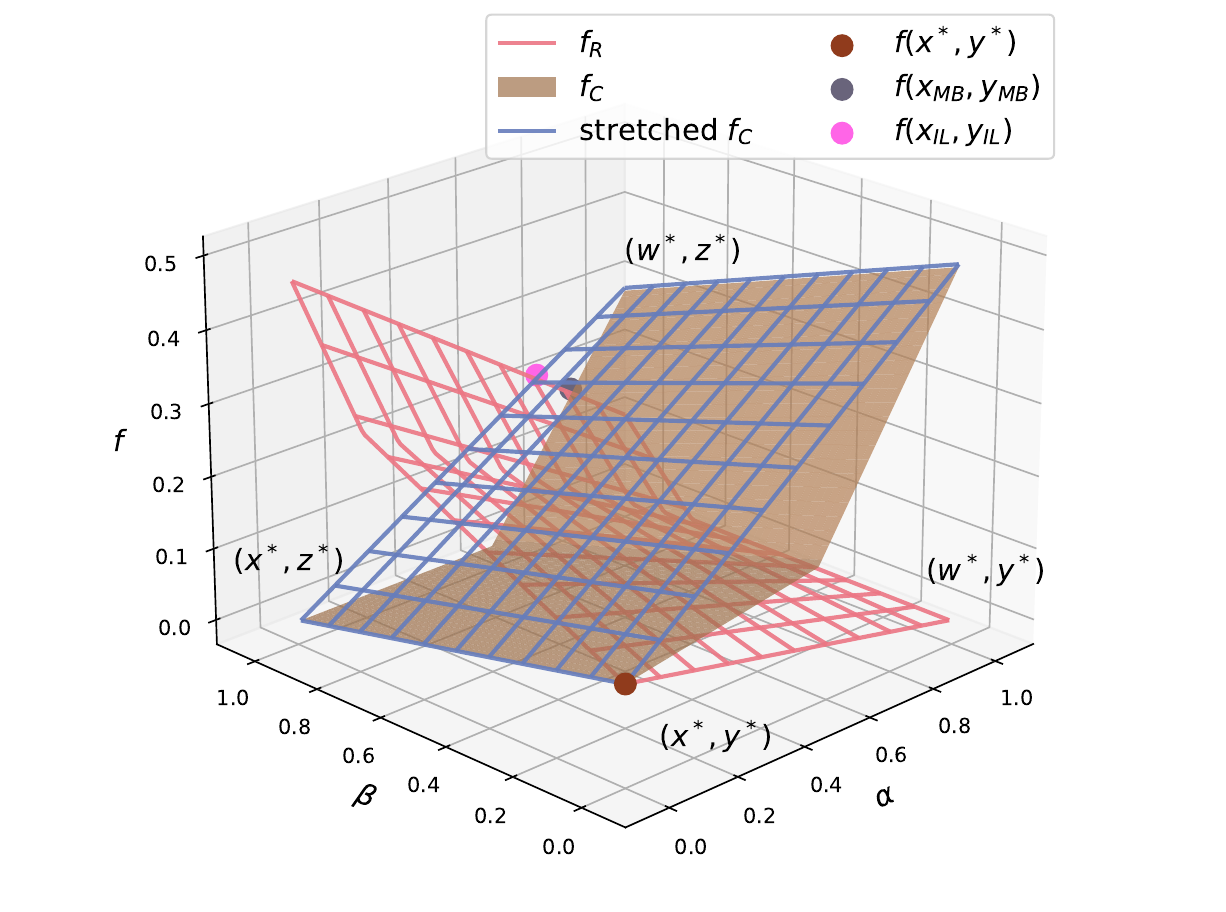}
    \caption{The procedure of stretching. Every section of $F_C$ when $\beta$ is fixed is stretched to a segment. Such a procedure lifts $\AdjSe$ to $\AdjTh$.}
    \label{fig:tostretch}
\end{figure}

\begin{figure}[ht]
    \centering
    \includegraphics[scale=0.55]{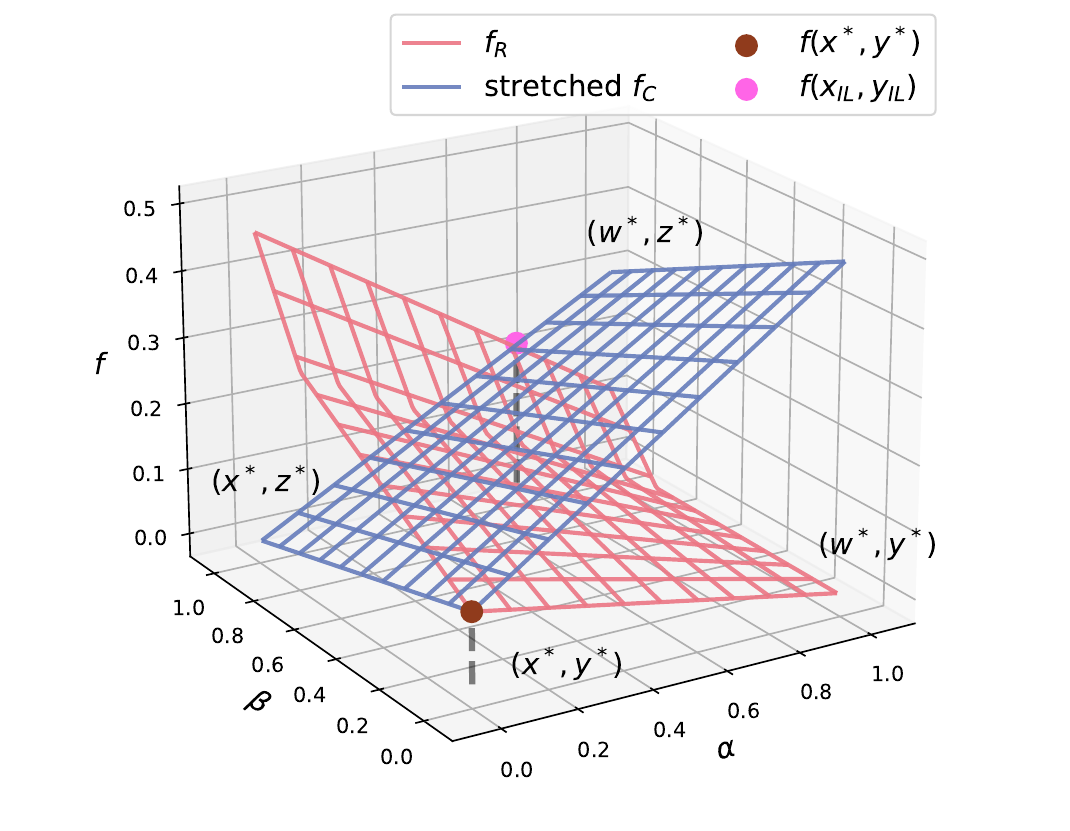}
    \caption{The result of stretching. Note that $\AdjSe$ is now lifted to $\AdjTh$.}
    \label{fig:stretched}
\end{figure}

Such an upper bound can be expressed in inequalities as well (we have presented them in \Cref{sec:algo} to define $\AdjTh$):
\begin{align*}
f_C(p w^*+(1-p)x^*,z^*)&\leq f_C(w^*,z^*)p,\\
f_R(w^*,q z^*+(1-q) y^*)&\leq f_R(w^*,z^*)q.
\end{align*}

After stretching, the original intersection point $\AdjSe$ shifts to a new point, which is exactly the definition of $\AdjTh$ in \Cref{sec:algo}. $\AdjTh$ appears easier to calculate and estimate. 

Let us leave the process of estimations for a while. Remember that we are doing \emph{worst-case} analysis, so an upper bound is not enough if it is never tight. Luckily, we have the following lemma and proposition that point out when such an upper bound becomes tight. From a geometric view, they present the equivalent condition that stretched figures and original figures are identical and that $\AdjSe$ coincides with $\AdjTh$.

\begin{lemma}\label{lemma:linearize}
The following two statements hold:
\begin{enumerate}
    \item $F_C(\alpha,\beta)=f_C(\alpha w^*+(1-\alpha)x^*,\beta z^*+(1-\beta)y^*)$ is a linear function of $\alpha$ if and only if
    \begin{align}\label{cond:linearize1}
        S_C(x^*)\cap S_C(w^*)\neq\varnothing.
    \end{align}
    \item $F_R(\alpha,\beta)=f_R(\alpha w^*+(1-\alpha)x^*,\beta z^*+(1-\beta)y^*)$ is a linear function of $\beta$ if and only if
    \begin{align}\label{cond:linearize2}
        S_R(y^*)\cap S_R(z^*)\neq\varnothing.
    \end{align}
\end{enumerate}
\end{lemma}

\begin{proof}
We only prove the first statement, and the second one is similar. Let $y_\beta=\beta z^*+(1-\beta)y^*$. Since $F_C(\alpha,\beta)$ is a convex function of $\alpha$, it suffices to prove that $Df_C(w^*,y_\beta,x^*,y_\beta)=-Df_C(x^*,y_\beta,w^*,y_\beta)$ if and only if \cref{cond:linearize1} holds. One can verify that
\begin{align*}
    Df_C(w^*,y_\beta,x^*,y_\beta)&=\max_{S_C(w^*)}(\T{C}x^*)-\T{x^*}Ry_\beta-\max(\T{C}w^*)+\T{w^*}Cy_\beta,\\
    Df_C(x^*,y_\beta,w^*,y_\beta)&=\max_{S_C(x^*)}(\T{C}w^*)-\T{w^*}Ry_\beta-\max(\T{C}x^*)+\T{x^*}Cy_\beta.
\end{align*}
Sum up these two equations and we have
\begin{align*}
    &Df_C(w^*,y_\beta,x^*,y_\beta)+Df_C(x^*,y_\beta,w^*,y_\beta)\\
    =&\max_{S_C(x^*)}(\T{C}w^*)-\max(\T{C}w^*)+\max_{S_C(w^*)}(\T{C}x^*)-\max(\T{C}x^*)\leq 0.
\end{align*}
and the equality holds if and only if $S_C(x^*)\cap S_C(w^*)\neq\varnothing$.
\end{proof}

\begin{proposition}\label{prop:stgy-comp}
$f\AdjFi\geq f\AdjSe$ and $f\AdjTh\geq f\AdjSe$ always hold. Meanwhile, $f\AdjSe=f\AdjTh$ holds if and only if
\begin{align*}
\begin{cases}
S_C(x^*)\cap S_C(w^*)\neq\varnothing,& \text{if } f_C(w^*,z^*)> f_R(w^*,z^*),\\
S_R(y^*)\cap S_R(z^*)\neq\varnothing,& \text{if } f_C(w^*,z^*)<f_R(w^*,z^*),\\
f_R(w^*,z^*)=f_C(w^*,z^*).
\end{cases}
\end{align*}
\end{proposition}
\begin{proof}
$f\AdjFi\geq f\AdjSe$ and $f\AdjTh\geq f\AdjSe$ are directly deducted by \Cref{lemma:best-on-boundary}. We now prove the second part. 

If $f_R(w^*,z^*)=f_C(w^*,z^*)$, then by \Cref{lemma:monotone-convex} and \Cref{lemma:best-on-boundary}, we obtain that $\AdjSe=\AdjTh=(w^*,z^*)$.

Suppose now $f_C(w^*,z^*)>f_R(w^*,z^*)$. Let 
\[x_\alpha:=\alpha w^*+(1-\alpha)x^*,\]
so $\AdjSe=(x_{\alpha^*},z^*)$. Notice that $f_R(x^*,z^*)\geq f_C(x^*,z^*)=0$, by intermediate value theorem and \Cref{lemma:monotone-convex}, the unique minimum point of $f$ on $\Gamma_2$, $\AdjSe$, lying on $\Phi=\{(x_\alpha,z^*):\alpha\in[0,1]\}$, is the intersection of $f_C$ and $f_R$. Again, by \Cref{lemma:monotone-convex}, $f_R$ is linear on $\Phi$ and $f_C$ is piecewise-linear on $\Phi$, therefore $\AdjSe$ coincides with $\AdjTh$ if and only if both $f_C$ is also linear on $\Phi$. By \Cref{lemma:linearize}, $f_C(x_\alpha,z^*)$ is linear on $\Phi$ if and only if $S_C(x^*)\cap S_C(w^*)\neq\varnothing$, which completes the proof of the case.

The case $f_R(w^*,z^*)>f_C(w^*,z^*)$ is symmetric, which we omit.
\end{proof}

We note that $\AdjFi$ is automatically included in \Cref{prop:stgy-comp}. Thus our analysis involves the adjustment in the original TS algorithm as well.

Now we turn back to estimations. We present the following estimations and inequalities for $f(x^*,y^*)$ and $f\AdjSe$ and show when the equality holds.

\begin{lemma}\label{lemma:est-for-stgy}
The following two estimations hold:
\begin{enumerate}
\item If $f_C(w^*,z^*)>f_R(w^*,z^*)$, then
\[
    f\AdjTh=\frac{f_R(x^*,z^*)(f_C(w^*,y^*)-\mu^*)}{f_C(w^*,y^*)+\lambda^*-\mu^*}\leq\frac{1-\mu^*}{1+\lambda^*-\mu^*}.
\]
And symmetrically, when $f_R(w^*,z^*)>f_C(w^*,z^*)$, we have
\[
    f\AdjTh=\frac{f_C(w^*,y^*)(f_R(x^*,z^*)-\lambda^*)}{f_R(x^*,z^*)+\mu^*-\lambda^*}\leq\frac{1-\lambda^*}{1+\mu^*-\lambda^*}.
\]
Furthermore, if $(x^*, y^*)$ is not a \NEn, the equality holds if and only if $f_C(w^*,y^*)=f_R(x^*,z^*)=1$.
\item $f(x^*,y^*)\leq \min\{\rho^*\lambda^*, (1 - \rho^*)\mu^*\}\leq\frac{\lambda^*\mu^*}{\lambda^*+\mu^*}$.
\end{enumerate}
\end{lemma}
\begin{proof}
The value of $f\AdjTh$ is obtained immediately by definition. We now show the inequality holds. We only prove the case when $f_C(w^*,z^*)>f_R(w^*,z^*)$ and the other case is symmetric. Notice that 
\begin{align*}
    f\AdjTh &= \frac{f_R(x^*,z^*)(f_C(w^*,y^*)-\mu^*)}{f_C(w^*,y^*)+\lambda^*-\mu^*} \\
    &\leq \frac{(f_C(w^*,y^*)-\mu^*)}{f_C(w^*,y^*)+\lambda^*-\mu^*} \\
    &\leq \frac{1 - \mu^*}{1 + \lambda^* - \mu^*}.
\end{align*}
The second line holds as $f\AdjTh \geq 0$ and $f_R(x^*, z^*)\leq 1$, and the third line holds as
\[G(t) = \frac{t}{t + \lambda^*}\]
is increasing on $(0, 1 - \mu^*]$. Moreover, by the proof of \Cref{lemma:lambda*-mu*}, $\lambda^* > 0$ as $f(x^*, y^*) > 0$. As a result, the equality holds if and only if $f_C(w^*, y^*) = f_R(x^*, z^*) = 1$.

For the second part, notice that 
\[f(x^*, y^*) = \min_{x', y'}T(x^*, y^*, x', y', \rho^*, w^*, z^*).\]
Therefore,
\begin{align*}
    f(x^*,y^*)&\leq T(x^*,y^*,x^*,z^*, \rho^*,w^*,z^*)=\rho^*\lambda^*,\\
    f(x^*,y^*)&\leq T(x^*,y^*,w^*,y^*,\rho^*,w^*,z^*)=(1-\rho^*)\mu^*,
\end{align*}
which immediately derives that $f(x^*,y^*)\leq \min\{\rho^*\lambda^*,(1-\rho^*)\mu^*\}\leq\lambda^*\mu^*/(\lambda^*+\mu^*)$.
\end{proof}

\begin{remark}\label{remark:1/2-sp}
\Cref{lemma:est-for-stgy} tells us that at worst a \SP could reach an approximation ratio of $1/2$. In fact, by the average value inequality, $f(x^*,y^*)\leq\lambda^*\mu^*/(\lambda^*+\mu^*)\leq(\lambda^*+\mu^*)/4\leq 1/2$. We now give the following game to demonstrate this. Consider the payoff matrices:
\begin{align*}
      R=\begin{pmatrix}0.5&0\\1&1\end{pmatrix},\qquad
      C=\begin{pmatrix}0.5&1\\0&1\end{pmatrix}.
\end{align*}
One can verify by \Cref{prop:construct-sp} that $(\T{(1, 0)},\T{(1, 0)})$ is a \SP with dual solution $\rho^* = 1/2, w^* = z^* = \T{(0, 1)}$ and  $f(x^*,y^*)=1/2$. Therefore, merely a \SP itself cannot beat a straightforward algorithm given by \cite{daskalakis2009note}, which always finds a solution with an approximation ratio no greater than $1/2$.
\end{remark}

The following lemma gives a numerical bound of the estimations in \Cref{lemma:ineq-for-est} and the equivalent condition that the quality holds.
\begin{lemma}[\cite{tsaknakis2008optimization}]\label{lemma:ineq-for-est}
Let
\begin{align*}
    b=\max_{s,t\in[0,1]}\min\left\{\frac{st}{s + t},\frac{1-s}{1+t-s}\right\},
\end{align*}
Then $b\approx 0.339321$, which is attained exactly at $s=\mu_0\approx 0.582523$ and $t=\lambda_0\approx 0.812815$.
\end{lemma}

For now, all the preparations are finished. All conditions that lead to a tight instance are given.\footnote{Precisely, these conditions guarantee that if we only make adjustments on boundary $\partial\Lambda$, we will attain a tight bound $0.3393$. But it suffices for the original TS algorithm. }  After several trials, one can find a tight instance.

At last, we prove \Cref{thm:worst-case-exist} by verifying the tight instance \cref{ex:tight-bound} with \SP $x^*=y^*=\T{(1,0,0)}$ and dual solution $\rho^*=\mu_0/(\lambda_0+\mu_0)$, $w^*=z^*=\T{(0,0,1)}$. Note that the theorem also guarantees the bound $0.3393$ when we try to adjust on the rectangle $\Lambda$, not only on the boundary $\partial\Lambda$.

\begin{proof}[Proof of \Cref{thm:worst-case-exist}]
We prove the theorem by verifying game \cref{ex:tight-bound} with stationary point $x^*=y^*=\T{(1,0,0)}$ and dual solution $\rho^*=\mu_0/(\lambda_0+\mu_0)$, $w^*=z^*=\T{(0,0,1)}$. Let $x_\alpha=\alpha w^*+(1-\alpha)x^*$, $y_\beta=\beta z^*+(1-\beta)y^*$.
\begin{enumerate}[label=\textit{Step \arabic*.}, fullwidth, listparindent=\parindent]
\item  Verify that $(x^*,y^*)$ is a \SPn. We have $f_R(x^*,y^*)=f_C(x^*,y^*)=b$. 
\[A(\rho^*,y^*,z^*)=(-0.1\rho^*+(1-\rho^*)b)e_3,\] 
therefore $\{1\}=\supp(x^*)\subset\{1,2,3\}=\suppmin(A(\rho^*,y^*,z^*))$. Condition \cref{cond:sp-1} holds. Similarly, condition \cref{cond:sp-2} holds, and the former statement is proved \Cref{prop:construct-sp}.
\item Verify that $S_C(x^*)\cap S_C(w^*)\neq\varnothing$ and $f_C(w^*,z^*)>f_R(w^*,z^*)$. The latter can be checked by direction calculation. One can calculate that $S_C(x^*)=\{2,3\}$ and $S_C(w^*)=\{2\}$, therefore their intersection is $\{2\}\neq\varnothing$. Consequently, by \Cref{prop:stgy-comp}, $f\AdjTh=f\AdjSe$ and by \Cref{lemma:linearize}, $f_C(x_\alpha,y_\beta)$ is a linear function of $\alpha$.
\item Verify that $\lambda^*=\lambda_0$, $\mu^*=\mu_0$, $f_R(x^*,z^*)=f_C(w^*,y^*)=1$, and $f(x^*,y^*)=b$. One can check these claims by calculation. Then by \Cref{lemma:est-for-stgy} and \Cref{lemma:ineq-for-est}, $f\AdjTh = f\AdjSe = b$.
\item Verify that $b\leq f(\alpha w^*+(1-\alpha)x^*,\beta z^*+(1-\beta)y^*)$ for any $\alpha,\beta\in[0,1]$.

First, we do a verification similar to step 2: $S_R(y^*)=\{2,3\}$ and $S_R(w^*)=\{2\}$, therefore $S_R(y^*)\cap S_R(z^*)=\{2\}\neq\varnothing$, and $f_R(x_\alpha,y_\beta)$ is a linear function of $\beta$.

Since $f_I(x_\alpha,y_\beta)\ (I\in\{R, C\})$ is a linear function of $\alpha$ or $\beta$, we can calculate the minimum point $(x^o(\beta),y_\beta)$ of $f$ given specific $\beta$.
\begin{align*}
    f_R(x_\alpha,y_\beta)&=b+(1-b)\beta-(b + (\lambda_0 - b)\beta)\alpha,\\
    f_C(x_\alpha,y_\beta)&=b - b\beta+(1 - b + (b - \mu_0)\beta)\alpha.
\end{align*}
and $x^o(\beta)$ satisfies
\begin{align*}
    &f_R(x^o(\beta),y_\beta)=f_C(x^o(\beta),y_\beta)\\
    \Longleftrightarrow\ & x^o(\beta)= \frac{\beta}{1+(\lambda_0-\mu_0)\beta}w^* + \left(1 - \frac{\beta}{1+(\lambda_0-\mu_0)\beta}\right)y^*.
\end{align*}

Now let
\[
    g(\beta):=f(x^o(\beta),y_\beta)=b + (1-b)\beta-\frac{(b + (\lambda_0 - b)\beta)\beta}{1+(\lambda_0-\mu_0)\beta}.
\]
As $\lambda_0 > \mu_0$, to prove that $\min_\beta g(\beta)=b$, it is sufficient to show
\begin{align*}
    (1-b)\beta(1+(\lambda_0-\mu_0)\beta)-(b + (\lambda_0 - b)\beta)\beta\geq 0.
\end{align*}
Or equivalently,
\begin{align*}
    h(\beta):=(1-2b)\beta+(b(1+\mu_0-\lambda_0)-\mu_0)\beta^2\geq 0.
\end{align*}

Notice that $h(\beta)$ has a negative coefficient on the square term, therefore $h(\beta)$ is a concave function. Further, we have $h(0)=0$ and $h(1)=1-2b+b(1+\mu_0-\lambda)-\mu_0>0$. By concavity, $h(\beta)\geq \beta h(1)+(1-\beta)h(0)\geq 0$.
\end{enumerate}

Now we complete the proof.
\end{proof}

From the proof of \Cref{thm:worst-case-exist}, we obtain the following useful corollaries.
\begin{corollary}\label{cor:boundary-bad}
Suppose $f(x^*,y^*)=f\AdjTh=b$. If either of the following two statements holds:
\begin{enumerate}
\item $S_C(x^*)\cap S_C(w^*)\neq\varnothing$ and $f_C(w^*,z^*)>f_R(w^*,z^*)$,
\item $S_R(y^*)\cap S_R(z^*)\neq\varnothing$ and $f_R(w^*,z^*)>f_C(w^*,z^*)$,
\end{enumerate}
then for any $(x,y)$ on the boundary of $\Lambda$, $f(x,y)\geq b$.
\end{corollary}

\begin{corollary}\label{cor:convex-bad}
Suppose $f(x^*,y^*)=f\AdjTh=b$, $S_C(x^*)\cap S_C(w^*)\neq\varnothing$ and $S_R(y^*)\cap S_R(z^*)\neq\varnothing$. Then for any $\alpha,\beta\in[0,1]$, $f(\alpha w^*+(1-\alpha)x^*,\beta z^*+(1-\beta)y^*)\geq b$.
\end{corollary}

It is worth noting that the game with payoff matrices \cref{ex:tight-bound} has a pure \NE with $x = y = \T{(0, 1, 0)}$, and the \SP 
\[(x^*,y^*) = (\T{(1, 0, 0)}, \T{(1, 0, 0)})\]
is a strictly-dominated strategy pair. However, the supports of strategies forming a \NE never include strictly-dominated pure strategies! We can also construct lots of games that are able to attain the tight bound but own distinct characteristics. For instance, we can give a game with no dominant strategies but attain the tight bound. Some examples are listed in \Cref{app:MTI}. Such results suggest that \SPn s may not be an optimal concept (in theory) for a better calculation of approximate \NEn.
\section{Generating Tight Instances}\label{sec:gen-TI}

In \Cref{sec:TI}, we prove the existence of tight game instances. 
Furthermore, as our preparations suggest, we can mathematically profile \emph{all} games that are able to attain the tight bound.
In this section, we gather properties in the previous sections and present a generator for such games.
Using the generator, we can dig into the previous three approximate \NE algorithms and reveal the behavior of these algorithms and also the features of \SPn s.
\Cref{algo:gen} is the generator of tight instances, in which the inputs are arbitrary $(x^*, y^*), (w^*, z^*) \in \Delta_m \times \Delta_n$. The algorithm outputs games such that $(x^*, y^*)$ is a \SP and $(\rho^* = \lambda_0 / (\lambda_0 + \mu_0), w^*, z^*)$ is a corresponding dual solution, or outputs ``NO'' if there is no such game.

The main idea of the algorithm is as follows.
\Cref{prop:construct-sp} shows an easy-to-verify equivalent condition of the \SPn; and all additional equivalence conditions required by a tight instance are stated in \Cref{prop:stgy-comp}, \Cref{lemma:est-for-stgy} and \Cref{lemma:ineq-for-est}. All of these conditions form convex linear restrictions over $(R,C)$.
Therefore, if we enumerate all pairs of possible pure strategies in $S_R(z^*)$ and $S_C(w^*)$ respectively, whether there exists a tight instance solution becomes a linear programming problem.
\vspace{0.5em}

\begin{breakablealgorithm}
\renewcommand{\algorithmicrequire}{\textbf{Input}}
\newcommand{\OUTPUT}[1]{\textbf{Output} #1}
\renewcommand{\algorithmiccomment}[1]{// \textit{#1}}
\caption{Tight Instance Generator}
\label{algo:gen}
\begin{algorithmic}[1]
\REQUIRE $(x^*, y^*),(w^*, z^*) \in\Delta_m\times\Delta_n$.
\bigskip
\IF{$\supp(x^*)=\{1,2,\ldots,m\}$ \OR $\supp(y^*)=\{1,2,\ldots,n\}$}
\STATE{\OUTPUT{``NO"}}
\ENDIF
\STATE $\rho^* \leftarrow \mu_0/(\lambda_0+\mu_0)$.
\bigskip
\STATE \COMMENT{Enumerate $k\in S_R(z^*)$ and $l\in S_C(w^*)$.}
\FOR{$k\in\{1,\ldots,m\}\setminus\supp(x^*), l\in\{1,\ldots,n\}\setminus\supp(y^*)$}
\medskip
\STATE Solve a feasible $R=(r_{ij})_{m\times n},C=(c_{ij})_{m\times n}$ from the following LP with no objective function:
\STATE \qquad\qquad\COMMENT{basic requirements.}
\STATE \qquad\qquad$0\leq r_{ij},c_{ij}\leq 1$ for $i\in\{1,\ldots,m\},j\in\{1,\ldots,n\}$,
\STATE \qquad\qquad$\supp(w^*)\subset S_R(y^*)$, $\supp(z^*)\subset S_C(x^*)$,
\STATE \qquad\qquad$k\in S_R(z^*)$, $l\in S_C(w^*)$,
\medskip
\STATE \qquad\qquad\COMMENT{ensure $(x^*,y^*)$ is a \SPn.}
\STATE \qquad\qquad$\supp(x^*)\subset\suppmin(-\rho^* Ry^*+(1-\rho^*)(Cz^*-Cy^*))$, \label{algline:SP1}
\STATE \qquad\qquad$\supp(y^*)\subset\suppmin(\rho^*(\T{R}w^*-\T{R}x^*)-(1-\rho^*)\T{R}x^*)$,\label{algline:SP2}
\medskip
\STATE \qquad\qquad\COMMENT{ensure $f(x^*,y^*) = b$.}
\STATE \qquad\qquad$\T{(w^*-x^*)}Ry^*=\T{x^*}C(z^*-y^*)=b$,\label{algline:f-b}
\medskip
\STATE \qquad\qquad\COMMENT{ensure $f\AdjTh = b$.}
\STATE \qquad\qquad$\T{x^*}Rz^*=\T{w^*}Cy^*=0$,\label{algline:br}
\STATE \qquad\qquad$r_{kj}=1$ for $j\in\supp(z^*)$, $c_{il}=1$ for $i\in\supp(w^*)$,\label{algline:f1}
\STATE \qquad\qquad$\T{w^*}Rz^*=\lambda_0$, $\T{w^*}Cz^*=\mu_0$,\label{algline:lambda-mu}
\medskip
\STATE \qquad\qquad\COMMENT{ensure $f\AdjSe=f\AdjTh$.}
\STATE \qquad\qquad$l\in S_C(x^*)$.\label{algline:2-3}
\medskip
\IF{LP is feasible}
\STATE \OUTPUT{feasible solutions}
\ENDIF
\ENDFOR
\bigskip
\IF{LP is infeasible in any round}
\STATE \OUTPUT{``No''}
\ENDIF
\end{algorithmic}
\end{breakablealgorithm}

\begin{proposition}\label{prop:generator}
Given $(x^*, y^*), (w^*, z^*)\in\Delta_m\times\Delta_n$, all the feasible solutions of the LP in \Cref{algo:gen} are all the games $(R, C)$ satisfying
\begin{enumerate}
\item $(x^*,y^*)$ is a stationary point,
\item tuple $(\rho^*=\mu_0/(\lambda_0+\mu_0), w^*, z^*)$ is the dual solution\footnote{One can verify that the value of $\rho^*$ in the dual solution of any tight \SP has to be $\mu_0/(\lambda_0+\mu_0)$, by the second part of \Cref{lemma:est-for-stgy}.},
\item $f_C(w^*,z^*)>f_R(w^*,z^*)$, and
\item $f(x,y)\geq b$ for all $(x,y)$ on the boundary of $\Lambda$.
\end{enumerate}
if such a game exists, and the output is ``NO'' if no such game exists.
\end{proposition}
\begin{proof}
By \Cref{prop:construct-sp}, line \ref{algline:SP1} and line \ref{algline:SP2} together form an equivalent condition of the first two statements that $(x^*,y^*)$ is a \SP and $(\rho^*,w^*,z^*)$ is the corresponding dual solution.

Now we prove the last two statements. By \Cref{lemma:best-on-boundary}, it suffices to prove that the algorithm outputs all games such that $f_C(w^*,z^*)>f_R(w^*,z^*)$, $f(x^*,y^*)\geq b$ and $f\AdjSe\geq b$. By \Cref{lemma:est-for-stgy} and \Cref{lemma:ineq-for-est}, we already have 
\[\min\{f(x^*,y^*),f\AdjTh\}\leq b,\]
and the equality holds if and only if $f(x^*,y^*)=f\AdjTh=b$. By \Cref{prop:stgy-comp}, $f\AdjTh\geq f\AdjSe$, so it suffices to show that $f(x^*,y^*)=f\AdjSe=f\AdjTh=b$ and $f_C(w^*,z^*)>f_R(w^*,z^*)$.

Line \ref{algline:f-b} ensures that $f(x^*,y^*)=b$. Line \ref{algline:br} and line \ref{algline:f1} together ensure that
\[f_R(x^*,z^*)=f_C(w^*,y^*)=1.\] 
Line \ref{algline:br} and line \ref{algline:lambda-mu} together ensure that $\lambda^*=\lambda_0$ and $\mu^*=\mu_0$. By \Cref{lemma:est-for-stgy} and \Cref{lemma:ineq-for-est}, these are equivalent conditions such that $f\AdjTh=b$, and it naturally leads to $f_C(w^*,z^*)>f_R(w^*,z^*)$ by \Cref{lemma:est-for-stgy}.

At last, by \Cref{prop:stgy-comp}, line \ref{algline:2-3} is the equivalent condition such that \[f\AdjTh=f\AdjSe.\]
\end{proof}


For the sake of experiments, there are three main concerns of the generator we take into account.

First, sometimes we want to generate games such that the minimum value of $f$ on the entire $\Lambda$ is also $b\approx 0.3393$. By \Cref{cor:convex-bad}, it suffices to add a constraint $S_R(y^*)\cap S_R(z^*)\neq\varnothing$ to the LP in \Cref{algo:gen}. This is not a necessary condition though.

Second, the dual solution of the LP is usually not unique, and we cannot expect which dual solution the LP algorithm yields. \cite{mangasarian1978uniqueness} gives some methods to guarantee that the dual solution is unique. In practice, we simply make sure that $w^*$ and $z^*$ are pure strategies. The reason is that even if the dual solution is not unique, the simplex algorithm will end up with some optimal dual solution on a vertex, in which cases, $w^*$ and $z^*$ are often both pure strategies.

Third, all feasible LP solutions form a convex polyhedron, which indicates that the cardinality of solutions is a continuum. Hence we need a sampling method to generate tight instances. A simple approach is to set a random object function, and the LP algorithm will find different vertices of the convex polyhedron. Make convex combinations of these vertices, and the results are samples of tight instances.

\section{Tightness of the Deligkas-Fasoulakis-Markakis Algorithm}\label{sec:1/3-tight}

Very recently, the work by Deligkas, Fasoulakis, and Markakis \cite{DBLP:journals/corr/abs-2204-11525} provides a polynomial-time algorithm computing a $1/3$-approximate Nash equilibrium. The DFM algorithm is also based on the same descent procedure but equipped with a more complicated adjustment method by using convex combinations with additional best response strategies beyond square $\Lambda$. 
They prove that such an adjustment method yields an upper bound approximation ratio of $1/3$. In this section, based on techniques developed in \Cref{sec:algo} and \Cref{sec:TI}, we show that $1/3$ is also the lower bound of the DFM algorithm.

We first introduce the adjustment method of the DFM algorithm. Suppose that $(x^*,y^*)$ is a stationary point and the corresponding dual solution is $(\rho^*, w^*,z^*)$. Recall that in \Cref{sec:TI}, we define $\lambda^*=\T{(w^*-x^*)}Rz^*$ and $\mu^*=\T{w^*}C(z^*-y^*)$. The adjustment presented in \Cref{algo:1/3-adj} is divided into four cases. In the case that $1/2<\lambda^*\leq 2/3<\mu^*$ and its symmetric case, the adjustment is delicate.

\begin{breakablealgorithm}
\renewcommand{\algorithmicrequire}{\textbf{Input}}
\newcommand{\OUTPUT}[1]{\textbf{Output} #1}
\renewcommand{\algorithmiccomment}[1]{// \textit{#1}}
\caption{Adjustment Method in the DFM algorithm.}
\label{algo:1/3-adj}
\begin{algorithmic}[1]
\REQUIRE $(x^*, y^*),(w^*, z^*) \in\Delta_m\times\Delta_n$, $\lambda^*,\mu^*\in[0,1]$.
\bigskip
\IF{$\min\{\lambda^*,\mu^*\}\leq 1/2$ \OR $\max\{\lambda^*,\mu^*\}\leq 2/3$}\label{line:case1start}
\STATE{\OUTPUT{$(x^*,y^*)$}}
\ENDIF
\IF{$\min\{\lambda^*,\mu^*\}\geq 2/3$}
\STATE{\OUTPUT{$(w^*,z^*)$}}
\ENDIF\label{line:case2end}
\IF{$1/2<\lambda^*\leq 2/3<\mu^*$}\label{line:case3start}
\STATE{$\hat{y}\leftarrow (y^*+z^*)/2$.}
\STATE{Pick $\hat{w}\in\Delta_m$ such that $\supp(\hat{w})\subseteq\suppmax(R\hat{y})$. }
\STATE{$t_r\leftarrow \T{\hat{w}}R\hat{y}-\T{w^*}R\hat{y}$, $v_r\leftarrow \T{w^*}Ry^*-\T{\hat{w}}Ry^*$, $\hat{\mu}\leftarrow \T{\hat{w}}Cz^*-\T{\hat{w}}Cy^*$.}
\IF{$v_r+t_r\geq(\mu^*-\lambda^*)/2$ \AND $\hat{\mu}\geq\mu^*-v_r-t_r$}
\STATE{$\displaystyle\alpha\leftarrow\frac{2(v_r+t_r)-(\mu^*-\lambda^*)}{2(v_r+t_r)}$.}
\STATE{Take $(x',y')$ between $(x^*,y^*)$ and $(\alpha w^*+(1-\alpha)\hat{w}, z^*)$ minimizing $f(x',y')$.}
\STATE{\OUTPUT{$(x',y')$}}
\ELSE
\STATE{$\displaystyle\beta\leftarrow\frac{1-\mu^*/2-t_r}{1+\mu^*/2-\lambda^*-t_r}$.}
\STATE{Take $(x',y')$ between $(x^*,y^*)$ and $(w^*, (1-\beta)\hat{y}+\beta z^*)$ minimizing $f(x',y')$.}
\STATE{\OUTPUT{$(x',y')$}}
\ENDIF
\ELSE
\STATE{Do the procedure symmetric to the previous ``if'' case.}
\ENDIF\label{line:case4end}
\end{algorithmic}
\end{breakablealgorithm}

We then show tight instances of the DFM algorithm matching the upper bound of $1/3$. Notice that for the first two cases (line \ref{line:case1start}-\ref{line:case2end}), one can verify that the $1/3$ bound is achieved with the following game, modified from game~\Cref{ex:tight-bound}:
\begin{align}\label{ex:tight-bound-1/3}
    R=\begin{pmatrix}0&0&0\\1/3&1&1\\1/3&1/2&1/2\end{pmatrix},\qquad
C=\begin{pmatrix}0&1/3&1/3\\0&1&1\\0&1&1\end{pmatrix}.
\end{align}

Game \cref{ex:tight-bound-1/3} attains the tight bound $1/3$ of the DFM algorithm at the \SP $x^*=y^*=\T{(1,0,0)}$ with dual solution $\rho^*=2/3$, $w^*=z^*=\T{(0,0,1)}$.

In the rest of this section, we focus on the last two cases (line \ref{line:case3start}-\ref{line:case4end}), which are more sophisticated.  We prove that for arbitrarily small $\epsilon>0$, an approximation ratio of $1/3-\epsilon$ can be reached by some instances.\footnote{The analysis on the proof of the upper bound in \cite{DBLP:journals/corr/abs-2204-11525} suggests that if some instance reaches an approximation ratio of $1/3$ in the last two cases, then it must hold that $\lambda^*=\mu^*=2/3$. However, due to the boundary condition of line \ref{line:case1start}, such instance should terminate in the first case and never fall into the last two cases. Such contradiction implies that $1/3$ is not attainable in the last two cases.} Such instance family is presented in \Cref{ex:tight-bound-1/3-eps}. Again, it is a modification of game \Cref{ex:tight-bound}.
\begin{align}\label{ex:tight-bound-1/3-eps}
    R=\begin{pmatrix}0&0&0\\1/3&1&1\\1/3&2/3-\epsilon/2&2/3-\epsilon/2\end{pmatrix},\qquad
C=\begin{pmatrix}0&1/3-\epsilon&1/3-\epsilon\\0&1&2/3+\epsilon\\0&1&2/3+\epsilon\end{pmatrix}.
\end{align}

The DFM algorithm reaches an approximation ratio of $1/3-\gamma(\epsilon)$ with \SP $x^*=y^*=\T{(1,0,0)}$ and dual solution $\rho^*=1/2$, $w^*=z^*=\T{(0,0,1)}$, where $\gamma(\epsilon)>0$ and $\gamma(\epsilon)\to 0$ as $\epsilon\to 0$.

We verify that for this instance, \Cref{algo:1/3-adj} terminates in case 3 (line \ref{line:case3start}) and outputs a strategy profile with the claimed approximation ratio.

First, we use \Cref{prop:construct-sp} to check that $(x^*,y^*)$ is indeed a stationary point with dual solution $(\rho^*,w^*,z^*)$. The direct calculation shows that 
\[A(\rho^*,y^*,z^*)=\left(\frac{1-3\epsilon}{6},\frac{1+3\epsilon}{6},\frac{1+3\epsilon}{6}\right)^T\text{\quad and}\]
\[B(\rho^*,x^*,w^*)=\left(\frac{1}{6},\frac{1}{6}+\frac{\epsilon}{4},\frac{1}{6}+\frac{\epsilon}{4}\right)^T.\] 
Thus 
\[\supp(x^*)=\{1\}=\suppmin(A(\rho^*,y^*,z^*))\text{\quad and}\]
\[\supp(y^*)=\{1\}=\suppmin(B(\rho^*,x^*,w^*)).\]

Second, it can be shown that $\lambda^*=2/3-\epsilon/2$ and $
\mu^*=2/3+\epsilon$, thus the input of \Cref{algo:1/3-adj} is valid and it falls exactly into case 3 (line \ref{line:case3start}).

Third, we calculate values of variables in case 3. $\hat{y}=(1/2,0,1/2)^T$, $\hat{w}=(0,1,0)^T$, $t_r=1/6+\epsilon/4$, $v_r=0$, $\hat{\mu}=2/3+\epsilon$. Thus when $\epsilon \leq 1/3$, $v_r+t_r\geq(\mu^*-\lambda^*)/2$ and $\hat{\mu}\geq\mu^*-v_r-t_r$. So we need to calculate the first branch, that is, $\alpha=1-(9\epsilon)/(2+3\epsilon)$.

At last, it can be calculated that $f(x^*,y^*)=1/3$ and 
\[f(\alpha w^*+(1-\alpha)\hat{w},z^*)=\max\left\{\left(1-\frac{9\epsilon}{2+3\epsilon}\right)\left(\frac{1}{3}+\frac{\epsilon}{2}\right),\frac{1}{3}-\epsilon\right\}=\frac{1}{3}-\gamma(\epsilon).\]

As $\epsilon\to 0$, $f(\alpha w^*+(1-\alpha)\hat{w},z^*)$ is arbitrarily close to $1/3$, as desired.


\section{Experimental Analysis}\label{sec:experi}

In this section, we further explore the characteristics of the algorithms presented in \Cref{sec:algo} with the help of numerical experiments. Such empirical results may provide us with a deep understanding of the behavior of these algorithms, specifically, the behavior of \SPn s and the descent procedure. Furthermore, we are interested in the tight instance generator itself presented in \Cref{sec:gen-TI}, particularly, on the probability that the generator outputs an instance given random inputs. At last, we compare the algorithms with other approximate \NE algorithms, additionally showing the potentially implicit relationships among these different algorithms. 

Readers can refer to \Cref{app:ED} for the details of the experiments. We here list the key results and insights we gain from these experiments.
\begin{enumerate}
    \item Our studies on the behavior of algorithms presented in \Cref{sec:algo} show that even in a uniformly sampled tight game instance, it is almost impossible for a uniformly-picked initial strategy pair to fall into the tight \SP at the termination. Such results suggest that uniform initialization leads to the dramatic inconsistency of tight instances of \SP algorithms between theory and practice.
    \item We then study the stability of tight \SPn s. A \SP $(x^*,y^*)$ is \emph{stable} if that, when we arbitrarily make a slight perturbation on $(x^*,y^*)$ and run the TS algorithm again, the algorithm generally terminates near $(x^*,y^*)$. We explore the stability on randomly generated tight instances with different sizes. In experiments, most tight instances of large sizes are not stable. Moreover, with the game size growing larger, the probability to find a stable tight instance becomes smaller and even vanishes. Thus it is really hard to meet an empirical approximation ratio of $0.3393$ in large-size games. Based on this result and further empirical studies, we give a time-saving and effective suggestion about the practical usage of the TS algorithm: \emph{If the algorithm terminates with a bad approximation ratio, slightly perturb the solution, and continue the algorithm. If the algorithm still terminates near the bad solution, randomly pick an initial point outside a small neighborhood of the solution, and rerun the algorithm.}
    \item Next, we turn to the tight instance generator described in \Cref{sec:gen-TI}. Given two arbitrary strategy pairs $(x^*, y^*)$ and $(w^*, z^*)$ in $\Delta_m\times \Delta_n$, we are interested in whether the generator outputs a tight game instance. The result shows that the intersecting proportion of $(x^*, y^*)$ and $(w^*, z^*)$ plays a vital role in whether a tight game instance can be successfully generated from these two pairs. It suggests that neither $x^\ast$ and $w^\ast$ share support, nor $y^\ast$ and $z^\ast$.
    \item At last, we measure how other algorithms behave on these tight game instances. Surprisingly, Czumaj et al.'s algorithm~\cite{czumaj2019distributed} terminates at an approximation ratio $b\approx 0.3393$ for all cases and all trials. Meanwhile, regret-matching algorithms~\cite{greenwald2006bounds} always find a pure \NE of a 2-player game if there exists, which is the case for all generated tight instances. Finally, \FP algorithm~\cite{brown1951iterative} behaves well on these instances, with a median approximation ratio of approximately $1\times 10^{-3}$ to $1.2\times 10^{-3}$ for games with different sizes.
\end{enumerate}

\section{Discussion}\label{sec:disc}
We present three problems that are expected to elicit a further understanding of \SPn s and the underlying structure of Nash equilibria.

\begin{enumerate}
    \item Analyze the dynamics of the descent procedure of the TS algorithm and provide stability analysis and smoothed analysis for worst cases.
\end{enumerate}

On both theoretical and experimental sides, we have yet to determine which kinds of \SPn s are easier to reach and which are not. It is also noticeable that a minor perturbation on the initial point leads to a significant difference in the convergence. All these phenomena are owing to the atypical behavior of the descent procedure. When we take these into consideration, the bound analysis becomes stability analysis and smoothed analysis.
\begin{enumerate}[resume]
\item Propose a benchmark for approximate \NE computing such that most existing polynomial-time algorithms have few advantages on the generated games.
\end{enumerate}
There is a natural extension to our tight instance generator: Find a class of games rendering the performances of most existing polynomial-time algorithms unsatisfying. It is worth noting that the classic game generator GAMUT~\cite{GAMUT} is conquered by the TS algorithm~\cite{fearnley2015empirical}: On games generated by GAMUT, the TS algorithm always finds a solution with an approximation ratio far better than $0.3393$. Therefore, a new benchmark is required, which is of great significance to understanding the hardness of \NE computing.

\begin{enumerate}[resume]
\item Propose a novel solution concept that calculates an $\epsilon$-approximate \NE directly without any further adjustment.
\end{enumerate}

The ultimate goal is to improve the approximation ratio. We summarize that all non-trivial polynomial-time approximation algorithms presented up till now involve two steps: first, to find a polynomial-time-solvable concept (usually by linear programming), and second, to make an adjustment step if the concept has an unsatisfying approximation ratio~\cite{bosse2010new,czumaj2019distributed,daskalakis2007progress,DBLP:journals/corr/abs-2204-11525,tsaknakis2008optimization}. The real challenge here is to propose a novel concept that characterizes the $\epsilon$-approximate Nash equilibria directly without any adjustment, which could show some insightful unknown structures of approximate Nash equilibria.

\bibliographystyle{plain}
\bibliography{ref}

\appendix
\section{Missing Calculations in the Main Body}\label{app:MissingCalculations}
\subsection{Calculating Derivatives}
We used explicit forms of $Df$, $Df_R$ and $Df_C$ in \Cref{sec:algo} but omitted the calculations of them. Since these calculations are rather complex, we present the detailed calculations here for completeness.

We first calculate the explicit form of $Df_R(x,y,x',y')$. The main calculation is the directional derivative of $\max(Ry)$ with respect to $y$. Let $h:=y'-y$. Notice that all entries of $\max(Ry)$ are continuous in $y$. Therefore, for sufficient small $\theta>0$, 
\[\max(R(y+\theta h))=\max_{S_R(y)}(R(y+\theta h)).\]

Since all entries of $Ry$ over $S_R(y)$ are equal (called property $(*)$), for sufficient small $\theta>0$, we have
\begin{align*}
    &\max(R(y+\theta h))-\max(Ry)\\
    =&\max_{S_R(y)}(R(y+\theta h))-\max_{S_R(y)}(Ry)\\
        =&\max_{S_R(y)}(R(y+\theta h)-Ry)\text{\quad (by property $(*)$)}\\
        =&\max_{S_R(y)}(\theta Rh)\\
        =&\theta\max_{S_R(y)}(R(y'-y))\\
        =&\theta\max_{S_R(y)}(Ry'-Ry)\\
        =&\theta\left(\max_{S_R(y)}(Ry')-\max_{S_R(y)}(Ry)\right)\text{\quad (by property $(*)$)}\\
        =&\theta\left(\max_{S_R(y)}(Ry')-\max(Ry)\right).
\end{align*}
As a result, we have
\[\lim_{\theta\to 0+}\frac{1}{\theta}\left(\max(R(y+\theta(y'-y)))-\max(Ry)\right)=\max_{S_R(y)}(Ry')-\max(Ry).\]
By basic calculus, we have
\begin{align*}
    &\lim_{\theta\to 0+}\frac{1}{\theta} \left(\T{(x+\theta(x'-x))}R(y+\theta(y'-y))-\T{x}Ry\right)\\
    =&\lim_{\theta\to 0+}\frac{1}{\theta} \left(\T{(\theta(x'-x))}Ry+\T{x}R(\theta(y'-y)+\T{(\theta(x'-x))}R(\theta(y'-y))\right)\\
    =&\T{(x'-x)}Ry+\T{x}R(y'-y).
\end{align*}
Combining these results, we get the following formula
\begin{align*}
    Df_R(x,y,x',y')&=\max_{S_R(y)}(Ry')-\max(Ry)-\T{(x'-x)}Ry-\T{x}R(y'-y)\\
    &=\max_{S_R(y)}(Ry')-\T{x'}Ry-\T{x}Ry'+\T{x}Ry-f_R(x,y).
\end{align*}
Similarly, we have
\begin{align*}
    Df_C(x,y,x',y')&=\max_{S_C(x)}(\T{C}x')-\max(\T{C}y)-\T{(x'-x)}Cy-\T{x}C(y'-y)\\
    &=\max_{S_C(x)}(\T{C}x')-\T{x'}Cy-\T{x}Cy'+\T{x}Cy-f_C(x,y).
\end{align*}

Now we calculate the derivative $Df(x,y,x',y')$. First, we consider the case that $f_R(x,y)\neq f_C(x,y)$. By continuity of $f_R$ and $f_C$, if $f_R(x_0,y_0)>f_C(x_0,y_0)$, then $f(x,y)=f_R(x,y)$ in some neighborhood of $(x_0,y_0)$; and if $f_R(x_0,y_0)<f_C(x_0,y_0)$, then $f(x,y)=f_C(x,y)$ in some neighborhood of $(x_0,y_0)$. Consequently,
\[Df(x,y,x',y')=\begin{cases}
Df_R(x,y,x',y'),& f_R(x,y)>f_C(x,y),\\
Df_C(x,y,x',y'),& f_R(x,y)<f_C(x,y).
\end{cases}\]

Finally, we calculate the derivative $Df(x,y,x',y')$ under the constraint that $f_R(x,y)=f_C(x,y)$, which is rather difficult by direct calculations. We develop the following lemma to handle it.

\begin{lemma}\label{lemma:cal-max-der}
Let $g, h$ be functions from $[0,1]$ to $\R$. $g, h$ are differentiable in the positive direction at $0$, and $g(0)=h(0)$. Then $\phi(x)=\max\{g(x),h(x)\}$ is differentiable in the positive direction at $0$, and
\[\phi'_+(0):=\lim_{\theta\to 0+}\frac{\phi(\theta)-\phi(0)}{\theta}=\max\{g'_+(0),h'_+(0)\}.\]
\end{lemma}
\begin{proof}
Without loss of generality, suppose that $f(0)=g(0)=0$ and $g'_+(0)\geq h'_+(0)$. By the result from analysis,
\begin{align*}
    g(x)&=g'_+(0)x+\alpha_1(x),\\
    h(x)&=h'_+(0)x+\alpha_2(x),
\end{align*}
where $\alpha_1(x)=o(x)$, $\alpha_2(x)=o(x)$ when $x\to 0^+$. It suffices to prove that
\[\phi(x)=g'_+(0)x+o(x).\]

If $g'_+(0)=h'_+(0)$, $\max\{g(x),h(x)\}=g'_+(0)x+\max\{\alpha_1(x),\alpha_2(x)\}$. Since both $\alpha_1(x)$ and $\alpha_2(x)$ are $o(x)$, clearly $\max\{\alpha_1(x),\alpha_2(x)\}=o(x)$.

If $g'_+(0)>h'_+(0)$, then for sufficient small $x_0>0$, $(g'_+(0)-h'_+(0))x+\alpha_1(x)-\alpha_2(x)>0$ holds for all $x\in(0,x_0)$, therefore $\phi(x)=g(x)=g'_+(0)x+o(x)$ when $x\to 0^+$.
\end{proof}

By \Cref{lemma:cal-max-der}, we obtain that
\begin{align*}
    Df(x,y,x',y')&=\max\{Df_R(x,y,x',y'),Df_C(x,y,x',y')\}\\
    &=\max\{T_1(x,y,x',y'),T_2(x,y,x',y')\}-f(x,y),
\end{align*}
where
\begin{align*}
T_1(x,y,x',y')&=\max_{S_R(y)}(Ry')-\T{x'}Ry-\T{x}Ry'+\T{x}Ry,\\ T_2(x,y,x',y')&=\max_{S_C(x)}(\T{C}x')-\T{x'}Cy-\T{x}Cy'+\T{x}Cy.
\end{align*}

\subsection{Calculating the Bilinear Form of \texorpdfstring{$T$}{T} and Dual Solution \texorpdfstring{$(\rho_0,w_0,z_0)$}{(rho0, w0, z0)} from Dual LP}
In this part, we give the detailed variations on $T$ that make $T$ a bilinear form and convert the problem $\min_{x',y'}\max_{\rho,w,z} T$ to the dual LP form. In \Cref{sec:prelim}, we denote an $n$-dimensional column vector with all entries equal to $1$ by $e_n$. We will use this notation in the calculations below. Recall the definition of $T$ in \Cref{sec:algo} that
\begin{align*}
T(x,y,x',y',\rho,w,z)=&\rho(\T{w}Ry'-\T{x}Ry'-\T{x'}Ry+\T{x}Ry)\\
         &\quad+(1-\rho)(\T{x'}Cz-\T{x}Cy'-\T{x'}Cy+\T{x}Cy).
\end{align*}
We have the following identities.
\begin{align*}
&\rho(\T{w}Ry'-\T{x}Ry'-\T{x'}Ry+\T{x}Ry)\\
=&\rho \T{w} Ry'-\rho(\T{w}e_m)x^TRy'-\rho(\T{w}e_m)\T{y}\T{R}x'+\rho(\T{w}e_m)(\T{e_m}x')\T{x}Ry\\
=&(\rho \T{w})(R-e_m\T{x}R)y'+(\rho \T{w})(-e_m\T{y}\T{R}+e_m\T{e_m}\T{x}Ry)x'.
\end{align*}
Similarly,
\begin{align*}
    &(1-\rho)(\T{x'}Cz-\T{x}Cy'-\T{x'}Cy+\T{x}Cy)\\
    =&((1-\rho)\T{z})(-e_n\T{x}C+e_n\T{e_n}\T{x}Cy)y'+((1-\rho)\T{z})(C-e_n\T{y}\T{C})x'.
\end{align*}
Combine these and we get
\begin{align*}
    &T(x,y,x',y',\rho,w,z)\\
    =&(\rho\T{w},(1-\rho)\T{z})\begin{pmatrix}R-e_m\T{x}R&-e_m\T{y}\T{R}+e_m\T{e_m}\T{x}Ry\\-e_n\T{x}C+e_n\T{e_n}\T{x}Cy&C-e_n\T{y}\T{C}\end{pmatrix}\begin{pmatrix}y'\\x'\end{pmatrix}\\
    =&(\rho\T{w},(1-\rho)\T{z})G(x,y)\begin{pmatrix}y'\\x'\end{pmatrix}.
\end{align*}

That is the bilinear form in the main body. To show that minimaxing this bilinear form produces a dual LP, we need the following calculations. Let $u_{i,n}$ denote an $n$-dimensional column vector whose $i$th entry is $1$ and the other entries are $0$. Let $0_n$ denote a $n$-dimensional zero column vector. Consider primal problem about variable $x', y'$:
\begin{gather*}
    \mathrm{minimize} \quad \max_{\substack{\rho\in[0,1],\\\supp(w)\subseteq S_R(y),\\\supp(z)\subseteq S_C(x),\\(w,z)\in\Delta_m\times\Delta_n}} T(x,y,x',y',\rho,w,z)\\
    \mathrm{s.t.}\quad (x',y')\in\Delta_m\times\Delta_n.
\end{gather*}
It is equivalent to the standard LP about variable $x',y',\delta_{1+}, \delta_{1-}$ below:
    \[\mathrm{minimize}~ \delta_{1+} - \delta_{1-}\]
    \begin{equation*}
        \begin{split}
        \mathrm{s.t.} ~-\begin{pmatrix}\T{u_{i,m}}&\T{u_{j,n}}\end{pmatrix}G(x, y) \begin{pmatrix} y' \\ x' \end{pmatrix} + \delta_{1+} - \delta_{1-} &\geq 0,~ i\in S_R(y), j\in S_C(x), \\
            \sum_{i=1}^n y'_i &\geq 1, \\
            \sum_{i=1}^n -y'_i&\geq -1, \\
            \sum_{i=1}^m x'_i &\geq 1, \\
            \sum_{i=1}^m -x'_i &\geq -1,\\
            x'_i &\geq 0,~ i=1,\ldots,m,\\
            y'_i &\geq 0,~ i=1,\ldots,n,\\
            \delta_{1+}&\geq 0,\\
            \delta_{1-}&\geq 0.
        \end{split}
    \end{equation*}

The dual problem about variable $\rho,w,z$ is
\begin{gather*}
    \mathrm{maximize} \quad \min_{(x',y')\in\Delta_m\times\Delta_n} T(x,y,x',y',\rho,w,z)\\
    \mathrm{s.t.}\quad (w,z)\in\Delta_m\times\Delta_n,\\
    \rho\in[0,1],\\
    \supp(w)\subseteq S_R(y),\\
    \supp(z)\subseteq S_C(x).
\end{gather*}
It is equivalent to the standard dual LP about variable $w',z',\delta_{21+},\delta_{21-},\delta_{22+},\delta_{22-}$ below:
\[\mathrm{maximize}~(\delta_{21+}-\delta_{21-})+(\delta_{22+}-\delta_{22-})\]
    \begin{equation*}
        \begin{split}
            \mathrm{s.t.}~-\begin{pmatrix} \T{u_{i,n}} & \T{0_m} \end{pmatrix} G(x, y)^T \begin{bmatrix} w' \\ z' \end{bmatrix} + \delta_{21+} - \delta_{21-} &\leq 0,~ i=1,\ldots,n, \\
            -\begin{pmatrix} \T{0_n} & \T{u_{i,m}} \end{pmatrix} G(x, y)^T \begin{bmatrix} w' \\ z' \end{bmatrix} + \delta_{22+} - \delta_{22-} &\leq 0,~ i=1,\ldots,m, \\
            w_i'&\leq 0,~ i\in\{1,\ldots,m\}\setminus S_R(y),\\
            z_i'&\leq 0,~ i\in\{1,\ldots,n\}\setminus S_C(x),\\
            \sum_{i=1}^m w_i' + \sum_{i=1}^n z_i' &\leq 1, \\
            -\sum_{i=1}^m w_i' - \sum_{i=1}^n z_i' &\leq -1,\\
            w_i'&\geq 0,~ i=1,\ldots,m,\\
            z_i'&\geq 0,~ i=1,\ldots,n,\\
            \delta_{21+}&\geq 0,\\
            \delta_{21-}&\geq 0,\\
            \delta_{22+}&\geq 0,\\
            \delta_{22-}&\geq 0.
        \end{split}
    \end{equation*}
Then the correspondence between dual solution $(\rho_0,w_0,z_0)$ and optimal solution $(w_0', z_0')$ of the dual LP above is $\rho_0=\sum_i w_{0i}'$, $w'=\rho_0 w_0$, and $z'=(1-\rho_0)z_0$. It is clear that these two LPs are the corresponding primal and dual LPs in the standard form.
\newpage
\section{Descent Procedure}\label{app:DP}

In this section, we show how to find a $\delta$-\SP in polynomial time of precision $\delta$. Recall that $V(x,y)$ is the min-max value of $T$ and $V(x,y)\leq f(x,y)$ always holds. Then immediately, we have
\begin{lemma}
    A strategy pair $(x,y)$ is a \emph{$\delta-$\SPn}, if and only if \[f_R(x,y)=f_C(x,y)\]
    and 
    \[V(x,y)-f(x,y)\geq -\delta.\]
\end{lemma}
When $\delta=0$, we have $V(x,y)=f_R(x,y)=f_C(x,y)$, and by \Cref{prop:sp-equi}, $(x,y)$ is a \SPn. 

Now we state the descent procedure given by \cite{tsaknakis2008performance}. The descent procedure is partitioned into 2 steps: seeking descent direction and line search.

Seeking descent direction, as the name suggests, is to find the steepest direction of $f$ in the sense of Dini directional derivative. We first fix one of $x,y$ and adjust the other to make $f_R=f_C$. Then we calculate the min-max value $V$ of $T$ and the steepest direction $(x' - x,y' - y)$ that attains $V$.

Having obtained the steepest direction $(x' - x,y' - y)$, we now search for a proper step size in that direction. This step is called line search in literature.

The pseudo-code of the descent procedure is presented in \Cref{algo:dp}.
\vspace{0.5em}
\begin{breakablealgorithm}
\renewcommand{\algorithmicrequire}{\textbf{Input}}
\newcommand{\OUTPUT}[1]{\textbf{Output} #1}
\newcommand{\GOTO}[1]{\textbf{goto} #1}
\renewcommand{\algorithmiccomment}[1]{// \textit{#1}}
\caption{Searching for a $\delta-$\SPn}
\label{algo:dp}
\begin{algorithmic}[1]
\REQUIRE Precision $\delta$, payoff matrices $R_{m\times n},C_{m\times n}$ and initial strategy $(x,y)\in\Delta_m\times\Delta_n$.
\bigskip
\STATE \COMMENT{seeking descent direction.}
\IF{$f_R(x,y)\neq f_C(x,y)$}\label{algline:DP1}
\IF{$f_R(x,y)>f_C(x,y)$}
\STATE Fix $y$ and solve the following LP with respect to $x$:
\begin{align*}
    &\min\limits_{x}\{\max(Ry)-\T{x}Ry\}, \\
    \text{s.t. }&f_R(x,y)\geq f_C(x,y), \\
    & x\in\Delta_m.
\end{align*} 
\ENDIF
\medskip
\IF{$f_R(x,y)<f_C(x,y)$}
\STATE Fix $x$ and solve the following LP with respect to $y$:
\begin{align*}
    &\min\limits_{y}\{\max(\T{C}x)-\T{x}Cy\}, \\
    \text{s.t. }&f_C(x,y)\geq f_R(x,y), \\
    & y\in\Delta_n.
\end{align*}

\ENDIF
\ENDIF
\medskip
\STATE Solve the following minimax problem with respect to $(x',y')$:
\begin{align*}
    &\min_{x',y'}\max_{\rho,w,z} T(x,y,x',y',\rho,w,z), \\
    \text{s.t. }& (x',y')\in\Delta_m\times\Delta_n, \\
    & \rho\in[0,1],\ \supp(w)\subset S_R(y)\text{ and }\supp(z)\subset S_C(x).
\end{align*} \label{algline:minimax}
\STATE Let the optimal value of LP be $V$ and the corresponding solution be $(x',y')$ and $(\rho, w, z)$.
\medskip
\IF{$V-f(x,y)\geq-\delta$}
\STATE \OUTPUT{$(x,y)$}\RETURN
\ENDIF
\bigskip
\STATE \COMMENT{line search.}
\STATE $\bar{S}_R(y)\leftarrow\{1,2,\ldots,m\}\setminus S_R(y)$, $\bar{S}_C(x)\leftarrow\{1,2,\ldots,n\}\setminus S_C(x)$.
\STATE 
\begin{align*}
    &\epsilon_1^*\leftarrow\min_{i\in\bar{S}_R(y)}\left(\frac{\max(Ry)-(Ry)_i}{\max(Ry)-(Ry)_i+(Ry')_i-\max_{S_R(y)}(Ry')}\right), \\
    &\epsilon_2^*\leftarrow\min_{j\in\bar{S}_C(x)}\left(\frac{\max(\T{C}x)-(\T{C}x)_j}{\max(\T{C}x)-(\T{C}x)_j+(\T{C}x')_j-\max_{S_C(x)}(\T{C}x')}\right), \\
    &\epsilon^*\leftarrow\min\{\epsilon_1^*,\epsilon_2^*,1\}.
\end{align*}
\STATE $H\leftarrow\min\{\T{(x'-x)}R(y'-y),\T{(x'-x)}C(y'-y)\}$.
\medskip
\IF{$H<0$}
\STATE $\epsilon^*\leftarrow\min\{\epsilon^*,|V-f(x,y)|/(2|H|)\}$.
\ENDIF
\medskip
\STATE $x\leftarrow x+\epsilon^*(x'-x)$, $y\leftarrow y+\epsilon^*(y'-y)$.
\STATE \GOTO{line \ref{algline:DP1}}
\end{algorithmic}
\end{breakablealgorithm}
\vspace{0.5em}

The minimax problem at line \ref{algline:minimax} can be solved by using the dual LP in \Cref{app:MissingCalculations}. We have the following convergence result.
\begin{theorem}[\cite{tsaknakis2008optimization}]
\Cref{algo:dp} terminates with a $\delta-$\SP in $O(\delta^{-2})$ steps for any $\delta>0$. Thus \Cref{algo:dp} finds a $\delta-$\SP in time $O(\delta^{-2}\poly(m,n))$.
\end{theorem}
\newpage
\section{Details of Experiments}\label{app:ED}
Throughout this section, we consider the distance induced by $L^\infty$ norm in $\Delta_m\times\Delta_n$.
\subsection{Behavior of the Stationary Point Algorithms}
In the very first experiment, we are quite interested in the behavior of algorithms we present in \Cref{sec:algo}. Specifically, given a tight game instance, we care much on the probability that these algorithms reach the tight bound $b\approx 0.3393$ with respect to the random choice of initial strategies. By the convexity of function $f(x,y)$, we can obtain the optimal adjustment by a ternary search algorithm. Therefore, this adjustment gives a lower bound for all convex combinations in square $\Lambda$.

We generate 20 games of size $3\times 3$, 15 games of size $3\times 4$, 10 games of size $4\times 4$, 3 games of size $5\times 5$ and one game of size $6\times 6$ by \Cref{algo:gen}, with respect to a random choice of $(x^*, y^*)$ and $(w^*, z^*)$ in $\Delta_m\times \Delta_n$. For $3\times 3$, $3\times 4$, $4\times 4$, $5\times 5$ and $6\times 6$ games, we partition the total space $\Delta_m\times\Delta_n$ into lattices of side length $1/10$, $1/10$, $1/8$, $1/6$ and $1/5$ respectively, and uniformly sample an initial strategy pair from each of them. \Cref{tab:init_test_approx} shows the behavior of the algorithm with the optimal adjustment on the approximation ratio, given the sampled initial strategies\footnote{We count the cases without distinguishing the specific games of the same size, as the result is similar for every game of the same size.}. It turns out that among all test cases, there are only 2 cases in one game with size $3\times 4$ that finally stop with $f>0.339$, while all other cases terminate with $f\approx 0$.

\begin{table}[ht]
    \caption{Statistics on approximation ratios on games with different sizes. The statistics is obtained by the algorithm with the optimal adjustment in \Cref{sec:algo}.}
    \centering
    \renewcommand{\arraystretch}{1.25}
    \begin{tabular}{|c|c|c|c|c|}
    \hline
    Game Size&$\#[\text{Sampled Points}]$&$\#[f>0.01]$& $\#[f>0.339]$&$\Pr[f>0.339]$\\\hline
    $3\times 3$&$200,000$&$0$&$0$&$0$\\\hline
    $3\times 4$&$1,500,000$&$2$&$2$&$1.3\times 10^{-6}$\\\hline
    $4\times 4$&$2,621,440$&$0$&$0$&$0$\\\hline
    $5\times 5$&$5,038,848$&$0$&$0$&$0$\\\hline
    $6\times 6$&$9,765,625$&$0$&$0$&$0$\\\hline
    \end{tabular}
    \label{tab:init_test_approx}
\end{table}

The result provides us with the following insight: even though we can promise the existence of a tight \SP by \Cref{algo:gen}, we cannot promise a high probability to actually find them in practice! This result roughly implies the inconsistency of tight instances of \SP algorithms between theory and practice.

\subsection{Stability of Tight Stationary Points}\label{sec:stability}
A tight stationary point is hard to find in practice, implying that even if we start the descent procedure near the tight \SPn, the TS algorithm may terminate at a faraway solution with a better approximation ratio. We call a \SP \emph{stable} if, under most slight perturbations, the TS algorithm will ultimately fall back to the same \SPn; Otherwise, we call it \emph{unstable}.

The stability in experiments interprets as follows. Let $(x^*,y^*)$ be a tight \SPn. Choose a ball centered at $(x^*,y^*)$ with radius $r$ (called perturbation ball). Randomly pick $s$ points in the ball and run the TS algorithm with every picked point as the initial point. If the algorithm terminates with a solution whose distance to $(x^*,y^*)$ is less than $r$, we call event ``fall-back" occurs. We say $(x^*,y^*)$ is \emph{stable} if ``fall-back" always occurs for any picked point.

Specifically, randomly generate $1,000$ games of size $3\times 3$, $600$ games of size $5\times 5$, $300$ games of size $7\times 7$, $150$ games of size $10\times 10$, $80$ games of size $15\times 15$, and $50$ games of size $20\times 20$. For a game of size $m\times n$, we randomly pick $s=16(m-1)(n-1)$ points in the perturbation ball with radius $r=0.01$. To sample games and perturbed points, we use different random methods as in \Cref{subsect:random-sample}. In the decent procedure of the TS algorithm, we choose precision parameter $\delta=0.001$. For each trial, we do the above perturbation experiment and count the number of ``fall-back"s. We count the number of tight \SPn s that are stable for each size of the games. The result is presented in \Cref{tab:perturb}.
\begin{table}[ht]
    \centering
    \caption{Statistics of stable tight \SPn s of different game sizes. We also calculate the probability that a randomly generated tight instance has a stable tight \SPn. By viewing the trials as independent Bernoulli trials, we calculate a $95\%$ confidence interval of the probability to be stable using Wilson's score method.}
    \label{tab:perturb}
    \renewcommand{\arraystretch}{1.25}
    \begin{tabular}{|c|c|c|c|c|}
    \hline
         Game Size&\#[trials]&\#[stable]&$\Pr[\text{stable}]$ ($95\%$ CI)  \\\hline
         $3\times3$&1,000&752&0.752 (0.724-0.778)\\\hline
         $5\times5$&600&101&0.168 (0.141-0.200)\\\hline
         $7\times7$&300&32&0.107 (0.077-0.147)\\\hline
         $10\times10$&150&0&0 (0-0.025)\\\hline
         $15\times15$&80&0&0 (0-0.046)\\\hline
         $20\times20$&50&0&0 (0-0.071)\\\hline
    \end{tabular}
\end{table}

As the table shows, in large games, most tight instances are not stable. Moreover, with the game size growing larger, the probability to find a tight instance with a stable tight \SP becomes smaller and even vanishes. Thus it is really hard to meet an empirical approximation ratio of $0.3393$.

\subsection{``Outside-the-Ball" Strategy}\label{sec:another-point}
As the previous part suggests, the TS algorithm performs quite well in practice. But what can we do if we, very unluckily, meet stable tight instances? We propose a strategy called ``outside-the-ball", i.e., randomly select an initial point outside the perturbation ball. Then we expect the algorithm would find better solutions.

We use the stable instances in the previous experiment. For each instance, randomly select $s'$ initial points outside the perturbation ball and run the TS algorithm with each initial point. To sample initial points, we use the random method as in \Cref{subsect:random-sample}. We say the ``outside-the-ball" strategy is \emph{effective} for one trial of one game if the TS algorithm terminates with a solution outside the perturbation ball and whose approximation ratio is below $0.339$. If in 95\% of the trials of a game, the ``outside-the-ball" strategy is effective, we also say the ``outside-the-ball" strategy is \emph{effective} for this game. The parameters to do the experiment are as follows. The perturbation ball has radius $r=0.01$. For a game of size $m\times n$, we have $s'=16(m-1)(n-1)$ initial points to try. We choose precision parameter $\delta=0.01$ in the decent procedure. The statistics of the experiment are shown in \Cref{tab:another-point}.
\begin{table}[ht]
    \centering
    \caption{Results of ``outside-the-ball" strategy. In fact, for any game, there are at most two trials in which ``outside-the-ball" strategy is not effective.}
    \label{tab:another-point}
    \renewcommand{\arraystretch}{1.25}
    \begin{tabular}{|c|c|c|}
         \hline
         Game Size&\#[stable]&\#[``outside-the-ball" effective]  \\\hline
         $3\times3$& 752 & 752\\\hline
         $5\times5$& 101 & 101\\\hline
         $7\times7$& 32 & 32\\\hline
    \end{tabular}
\end{table}

All trials of the ``outside-the-ball" strategy succeed! Thus it is a good strategy if we have bad luck to meet stable tight instances.

\subsection{Behavior of the Tight Instance Generator}\label{subsec:Exper-TIG}

In this part, we turn to the tight instance generator we described in \Cref{sec:gen-TI}. We focus on the ``efficiency'' of the tight instance generator, or formally, given two random strategy pairs $(x^*, y^*)$ and $(w^*, z^*)$ in $\Delta_m\times \Delta_n$, the probability that the generator outputs a tight game instance. As we have already shown in \Cref{prop:generator}, our generator can generally provide all tight game instances. Therefore, the results may bring us with further understanding on the distribution of tight \SPn s in a general sense. 

For the experiment setting, we consider games with $5$ different sizes from $3\times 3$ to $7\times 7$. Further, for the restriction on the \SP and its dual solution, we give four different conditions: (1) no restriction as the control group, (2) $\supp(x^*)\cap \supp(w^*) = \supp(y^*)\cap \supp(z^*) = \varnothing$, (3) $\supp(x^*)\cap \supp(w^*) \ne \varnothing$ and $\supp(y^*)\cap \supp(z^*) \ne \varnothing$, and (4) $\supp(w^*)\subseteq \supp(x^*)$ and $\supp(z^*)\subseteq \supp(y^*)$. Under each of the $5\times 4 = 20$ possible combinations of a game size and a restriction, we uniformly sample $1,000$ pairs of $(x^*, y^*)$ and $(w^*, z^*)$ and count the success rate that a tight game instance can be generated. The result is shown in \Cref{tab:SR-TIG}.
\begin{table}[t]
    \caption{Success rate of generating a tight game instance $(R, C)$ with different game size under different restrictions on $(x^*, y^*)$ and $(w^*, z^*)$.}
    \centering
    \renewcommand{\arraystretch}{1.25}
    \begin{tabular}{|c|c|c|c|c|c|}
    \hline
    \diagbox{Restrictions}{Success Rate}{Game Size} & $3\times 3$ & $4\times 4$ &$5\times 5$ &$6\times 6$ &$7\times 7$\\
    \hline
    No Restriction & 3.0\% & 5.0\% & 5.7\% & 6.8\% & 6.3\% \\
    \hline
    $\supp(x^*)\cap \supp(w^*) = \varnothing,$ & \multirow{2}*{56.4\%} & \multirow{2}*{85.6\%} & \multirow{2}*{94.7\%} & \multirow{2}*{98.0\%} & \multirow{2}*{99.3\%} \\
    $\supp(y^*)\cap \supp(z^*) = \varnothing$ & & & & & \\
    \hline
    $\supp(x^*)\cap \supp(w^*) \ne \varnothing,$ & \multirow{2}*{0.0\%} & \multirow{2}*{0.0\%} & \multirow{2}*{0.4\%} & \multirow{2}*{0.3\%} & \multirow{2}*{1.1\%} \\
    $\supp(y^*)\cap \supp(z^*) \ne \varnothing$ & & & & & \\
    \hline
    $\supp(w^*)\subseteq \supp(x^*),$ & \multirow{2}*{0.0\%} & \multirow{2}*{0.0\%} & \multirow{2}*{0.0\%} & \multirow{2}*{0.0\%} & \multirow{2}*{0.0\%} \\
    $\supp(z^*)\subseteq \supp(y^*)$ & & & & & \\
    \hline
    \end{tabular}
    \label{tab:SR-TIG}
\end{table}

\Cref{tab:SR-TIG} shows that when the supports of $x^*$, $w^*$ and those of $y^*$, $z^*$ do not intersect, respectively, the successful generating rate increases with the game size.  Meanwhile, when their support always intersect (Condition~(3)), the successful generating rates are all small under experimented game sizes. At last, surprisingly, we discover that when $\supp(w^*)\subseteq \supp(x^*)$ and $\supp(z^*)\subseteq \supp(y^*)$, the success rate remains zero whatever the game size is!

\subsection{Comparison with Other Algorithms}

\begin{figure}[t]
    \centering
    \includegraphics[width=\textwidth]{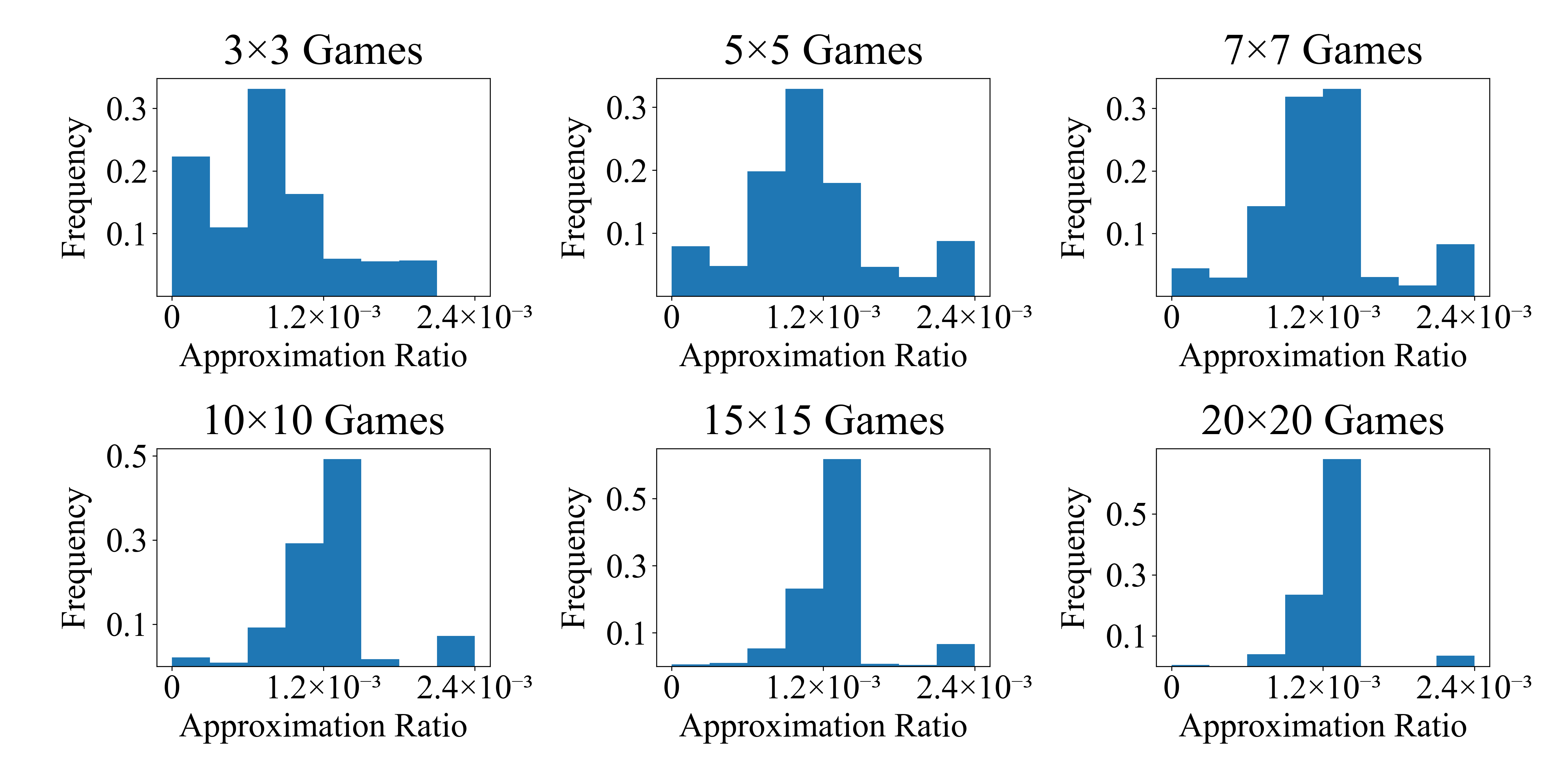}
    \caption{Performance of \FP algorithm on tight game instances.}
    \label{fig:FP}
\end{figure}

In the very last, we make an experiment on how other approximate \NE algorithms behave on those tight instances that \SP algorithms are expected not to perform well, and therefore compare \SP algorithms with these algorithms. Specifically, we consider three algorithms: Czumaj et al.'s algorithm~\cite{czumaj2019distributed} with an approximation ratio of $0.38$, regret-matching algorithms~\cite{greenwald2006bounds} in online learning, and \FP algorithm~\cite{brown1951iterative} with an approximation ratio of $1/2$ within constant rounds~\cite{conitzer2009approximation,feder2007approximating}.

We generate 1,860 tight game instances with different sizes by \Cref{algo:gen} for the test: 1,000 $3\times 3$ games, 500 $5\times 5$ games, 200 $7\times 7$ games, 100 $10\times 10$ games, 50 $15\times 15$ games and 10 $20\times 20$ games. For every game instance, we run each of the three algorithms $20$ times provided the randomness of these algorithms and count the approximation ratio.

The results turn out to be surprising. First, Czumaj et al.'s algorithm terminates at an approximation ratio of $0.3393$ for all cases and all trials. The reason for such a consequence is that the \NE of the zero-sum game specified by the algorithm already satisfies the required approximation ratio of $0.38$; therefore, the further adjustment step never happens. 

Meanwhile, regret-matching algorithms always find a pure \NE of a 2-player game if there exists one, which is the case for all generated tight instances. However, we believe that there still exists some tight game instances with no pure \NEn.

At last, \Cref{fig:FP} shows the performance of \FP algorithm on these tight instances.\footnote{Similar to the first experiment, we count the cases without distinguishing the specific games of the same size, since the result is similar for every game of the same size.} For games of all sizes, \FP algorithm shows a bell-shaped distribution on the approximation ratios, with median value approximately $1.2\times 10^{-3}$ on games with size no smaller than $5\times 5$. For $3\times 3$ games, \FP algorithm behaves even better, with lots of instances finding a \NEn, and median value decreasing to approximately $1.0\times 10^{-3}$. We explain the good performance of \FP algorithm as the set of tight instances of \FP algorithm are of zero measure over the set of tight instances of algorithms presented in \Cref{sec:algo}. In other words, there are some special latent relationships between \FP algorithm and stationary points.

\subsection{Random Methods for Sampling}\label{subsect:random-sample}
In this part, we introduce the omitted details of random samplings in previous experiments.
We employ different random methods for three objects: games in most experiments, perturbed points in \Cref{sec:stability}, and initial points of ``outside-the-ball" in \Cref{sec:another-point}.

\textbf{Games.}
To sample $t$ games of size $m\times n$, we need $10$ groups of $x^*,y^*,w^*,z^*$ that \Cref{algo:gen} can use to efficiently generate games (see \Cref{subsec:Exper-TIG}), and then sample $t/10$ tight instances for each $x^*,y^*,w^*,z^*$.
\begin{enumerate}[label=Step \arabic*.]
    \item Uniformly pick nonempty sets $\supp(x^*)$, $\supp(y^*)$, $\supp(w^*)$ and $\supp(z^*)$ satisfying \[\supp(x^*)\neq\{1,2\ldots,m\},\] \[\supp(y^*)\neq\{1,2\ldots,n\},\] \[\supp(x^*)\cap \supp(w^*) = \varnothing,\] \[\text{and }\supp(y^*)\cap \supp(z^*) = \varnothing.\] For each vector $x^*,y^*,w^*,z^*$, uniformly and independently pick real numbers from $[0,1]$ for the support indices. Normalize these vectors to $\Delta_m$ or $\Delta_n$. By this method, we obtain valid $x^*,y^*,w^*,z^*$. Input it to \Cref{algo:gen}.
    \item For the LP in each enumeration of $k\in S_R(z^*)$ and $l\in S_C(w^*)$ in \Cref{algo:gen}, uniformly and independently pick $m$ vectors in $[0,1]^{2mn}$ as the coefficients of the object function for the LP. And then for each function, choose to maximize or minimize it with equal probabilities. Let the optimal solutions for $m$ object functions be $(R_i,C_i)_{i=1}^m$. Uniformly and independently pick $t/10$ vectors from $[0,1]^{mn}\times[0,1]^{mn}$ and normalize them. Make $t/10$ convex combinations of $(R_i,C_i)$ using these vectors as weights and output the combination results.
\end{enumerate}

\textbf{Perturbed points.}
Suppose the radius of the perturbation ball is $r$. Uniformly and independently pick $m+n$ real numbers from $[-r,r]$. Add $m$ of them to $x^*$ and $n$ of them to $y^*$ index by index. Make the negative entries in the added vectors be $0$. Normalize the vectors and output them as a perturbed point.

\textbf{Initial points of ``outside-the-ball''.}
Uniformly and independently pick two vectors from $[0,1]^m$ and $[0,1]^n$, respectively. Normalize these two vectors. If the normalization result lies in the perturbation ball, redo the previous procedure. Otherwise, output the result as an initial point of the strategy ``outside-the-ball''.
\newpage
\section{More Tight Instances}\label{app:MTI}
We present two more tight instances without proof. One can check their correctness by similar steps in \Cref{thm:worst-case-exist}.

\begin{example}[Tight instances of every size]
Suppose $m,n>2$. 
The game with payoff matrices \cref{ex:every-size} attains the tight bound $b$ at \SP $x^*=\T{(1,0,\ldots,0)_m}$, $y^*=\T{(1,0,\ldots,0)_n}$ and dual solution $\rho^*=\mu_0/(\lambda_0+\mu_0)$, $w^*=\T{(0,1,0,\ldots,0)_m}$ and $z^*=\T{(0,1,0,\ldots,0)_n}$.
\begin{align}
    R&=\begin{pmatrix}
    0.1&0&\ldots&0\\
    0.1+b&\lambda_0&\ldots&\lambda_0\\
    0.1+b&1&\ldots&1\\
    \vdots&\vdots&\ddots&\vdots\\
    0.1+b&1&\ldots&1
    \end{pmatrix}_{m\times n},\notag\\
    C&=\begin{pmatrix}
    0.1&0.1+b&0.1+b&\ldots&0.1+b\\
    0&\mu_0&1&\ldots&1\\
    \vdots&\vdots&\vdots&\ddots&\vdots\\
    0&\mu_0&1&\ldots&1
    \end{pmatrix}_{m\times n}.\label{ex:every-size}
\end{align}
\end{example}

\begin{example}[Tight instance with no dominated strategy]
The game with payoff matrices \cref{ex:no-dominated} attains the tight bound $b$ at \SP $x^*=y^*=\T{(1/2,1/2,0,0)}$, and dual solution $\rho^*=\mu_0/(\lambda_0+\mu_0)$, $w^*=z^*=\T{(0,0,1/2,1/2)}$. One can verify that there is no dominated strategy for either player in this game.
\begin{align}
    R&=\begin{pmatrix}
    2b+0.2&0&0&0\\
    0&2b+0.2&0&0\\
    2b+0.17&2b+0.03&1&1\\
    2b+0.03&2b+0.17&2\lambda_0-1&2\lambda_0-1
    \end{pmatrix},\notag\\
    C&=\begin{pmatrix}
    2b+0.2&0&2b+0.17&2b+0.03\\
    0&2b+0.2&2b+0.03&2b+0.17\\
    0&0&1&2\mu_0-1\\
    0&0&1&2\mu_0-1
    \end{pmatrix}.\label{ex:no-dominated}
\end{align}
\end{example}

\end{document}